\documentclass[11pt]{article}

\usepackage{enumitem}
\usepackage[centertags]{amsmath}
\usepackage{amsfonts,amsthm,amssymb}
\usepackage{amssymb}
\usepackage{amsmath}
\usepackage{url}
\usepackage{soul}
\usepackage{graphicx}
\usepackage{accents}
\usepackage{booktabs}
\usepackage{appendix}
\usepackage{fullpage}
\usepackage{etoolbox}
\usepackage{xparse}
\usepackage{bm}
\usepackage{booktabs}
\usepackage{mathtools}
\usepackage{color}
\usepackage[svgnames]{xcolor}
\usepackage{soul}
\usepackage{chngcntr}
\usepackage{authblk}

\newcommand{\n}{\normalsize}
\newcommand{\s}{\scriptsize}
\newcommand{\ubar}[1]{\underaccent{\bar}{#1}}
\newcommand{\dr}{\rotatebox[origin=c]{180}{$\Lsh$}}
\DeclareDocumentCommand{\publicBelief}{O{\mu}}{#1}
\DeclareDocumentCommand{\privateBelief}{O{p}}{#1}
\DeclareDocumentCommand{\signalupdate}{O{q}}{#1}
\DeclareDocumentCommand{\signal}{O{s}}{#1}
\DeclareDocumentCommand{\signalsSet}{O{S}}{#1}
\DeclareDocumentCommand{\state}{O{\omega}}{#1}
\DeclareDocumentCommand{\statesSet}{O{\Omega}}{#1}
\DeclareDocumentCommand{\price}{O{\tau}}{#1}
\DeclareDocumentCommand{\action}{O{a}}{#1}
\DeclareDocumentCommand{\limitParam}{O{\alpha}}{#1}
\DeclareDocumentCommand{\deterrencePrice}{O{\price} O{d}}{#1^#2}
\DeclareDocumentCommand{\LLR}{O{x}}{\log(\frac{#1}{1-#1})}
\DeclareDocumentCommand{\lBound}{O{\limitParam} O{\publicBelief}}{\ubar{#1}_{#2}}
\DeclareDocumentCommand{\uBound}{O{\limitParam} O{\publicBelief}}{\bar{#1}_{#2}}
\DeclareDocumentCommand{\eqPrice}{O{\price}}{#1^{*}}

\newtheorem{lemma}{Lemma}
\newtheorem{proposition}{Proposition}

\newtheorem{theorem}{Theorem}

\newtheorem{conjecture}[theorem]{Conjecture}
\newtheorem{corollary}{Corollary}

\theoremstyle{definition}

\newlist{secenum}{enumerate}{10}
\setlist[secenum]{label=\thesection.\arabic*,leftmargin=*}

\newtheorem{innercustomgeneric}{\customgenericname}
\providecommand{\customgenericname}{}
\newcommand{\newcustomtheorem}[2]{%
	\newenvironment{#1}[1]
	{%
		\renewcommand\customgenericname{#2}%
		\renewcommand\theinnercustomgeneric{##1}%
		\innercustomgeneric
	}
	{\endinnercustomgeneric}
}

\newcustomtheorem{customthm}{Theorem}
\newcustomtheorem{customlemma}{Lemma}
\newcommand{\blocktheorem}[1]{%
	\csletcs{old#1}{#1}
	\csletcs{endold#1}{end#1}
	\RenewDocumentEnvironment{#1}{o}
	{\par\addvspace{1.5ex}
		\noindent\begin{minipage}{\textwidth}
			\IfNoValueTF{##1}
			{\csuse{old#1}}
			{\csuse{old#1}[##1]}}
		{\csuse{endold#1}
		\end{minipage}
		\par\addvspace{1.5ex}}
}

\blocktheorem{customlemma}

\usepackage{array}
\usepackage{makecell}

\definecolor{ao}{rgb}{0.0, 0.5, 0.0}

\usepackage{hyperref}
\usepackage{subcaption}
\usepackage{multirow}

\usepackage{colortbl}

\usepackage{caption}

 \usepackage{placeins}

\usepackage{amsmath}
\usepackage{amssymb}
\title{Mutual Insurance and Delegation in Sequential Fundraising}

\author[1]{Amir Ban\thanks{amirban@me.com and mkoren@cmsa.fas.harvard.edu }}
\author[2]{Moran Koren\thanks{Research supported by the Fulbright Postdoctoral Fellowship 2019/2020}}
\affil[1]{Weizmann Institute of Science}
\affil[2]{Harvard University}

\renewcommand\footnotemark{}

\date{\today}

\begin{document}

\maketitle
\begin{abstract}
Seed fundraising for ventures often takes place by sequentially approaching potential contributors, who make observable decisions. The fundraising succeeds when a target number of investments is reached. Though resembling classic information cascades models, its behavior is radically different, exhibiting surprising complexities. Assuming a common distribution for contributors' levels of information, we show that participants rely on {\em mutual insurance}, i.e., invest despite unfavorable information, trusting future player strategies to protect them from loss. {\em Delegation} occurs when contributors invest unconditionally, empowering the decision to future players. Often, all early contributors delegate, in effect empowering the last few contributors to decide the outcome. Similar dynamics hold in sequential voting, as in voting in committees.
\end{abstract}

\newpage

\section{Introduction}

Mom and Dad are asked to approve of Son's latest bright idea, whose wisdom is uncertain to both. Son will go ahead only if both parents approve. If the idea turns out well, each parent gains an amount $W$, while if it is bad, each loses an equal amount. If they do not approve, no one gains or loses.

The parents are fully rational, and each has an independent, private signal of the idea's wisdom, which is either ``good'' or ``bad'', and is called the {\em state of the world}. The {\em quality} of each signal is, for each parent, the probability that the signal matches the state of the world. The signal quality ranges from being uninformative about the state of the world, to deterministically revealing it. Each parent's quality is independently drawn from a known quality distribution, and is privately known by the respective parent.\footnote{The problem is unchanged if the privately-known qualities are merely beliefs. This is because the players are rational and do not entertain conscious biases.}

Suppose Dad is asked to make a decision first. His decision is observed by Mom, who decides last. What is the parents' equilibrium strategy, given a common prior for the state of the world?

This question appears like an elementary exercise in social choice, or in information cascades, or in crowdfunding, but it hides complexity and elements that are novel and absent from the aforementioned fields. Chief among them is {\em mutual insurance}, wherein an early agent (Dad, in this case) relaxes his criteria for participation relying on future agents' (Mom, in this case) strategies to protect him against loss. Thus, in the classical information cascades model of \cite{Banerjee1992} and of \cite{bikhchandani1992theory}, Dad would approve strictly based on having a private posterior probability for a good state of the world of at least $50\%$, as Mom's future decision is irrelevant for his. In our problem, Dad would participate even if his posterior is below  $50\%$. This is because even if, {\em prima facie}, he expects loss, he expects Mom's future decision to weed out more scenarios where he loses than where he gains. Therefore, his total expectation from participation is positive.

Another element hidden in our small problem we call {\em delegation}. This occurs when a player, usually the first (Dad), assumes so much mutual insurance that his optimal strategy is to participate for any signal or quality, making his decision  uninformative. This is not the same as the up-cascade feature of information cascades, where the prior is so favorable that all players participate regardless of type (although up- and down-cascades are also features of our problem). Rather, the prior is not favorable enough for both parents to approve automatically, but is suitable for Dad to delegate the decision to Mom, who, in this case, has to make a solo decision, as Dad's action carries no information.

Our Mom and Dad problem is the simplest case of a more general sequential decision problem, with $n$ players, the approval (now also called {\em investment}) of at least $B$ of them needed for a project to take place, with those who approved subject to gain or loss, depending on the project's outcome. Each player gets an independent good/bad signal on the project's prospects, with varying quality, independently drawn from a known common distribution, and privately known (or believed) by the player. Players act sequentially, with their actions, but not their signal or quality, publicly observed.\footnote{In a fundraising, information flows through the fundraiser (Son in our example). The importance contributors ascribe to information, and the balance of power between them and the fundraiser, motivates him to make it available to them.}

Our Mom and Dad problem represents $B = 2, n = 2$, while the classical model of  \cite{Banerjee1992} and \cite{bikhchandani1992theory} has $B \leq 1$. The case $B = n$ is a unanimity game, where a project takes place only if all players approve. The general case, $B \leq n$, captures sequential fundraising, where $n$ potential contributors are sequentially approached for a contribution, and the project takes place if at least $B$ of them contribute.

In sequential voting, participants vote sequentially ``aye'' or ``nay'' to a resolution, which passes if it receives a minimal number of ``aye'' votes. It occurs in committees, in legislatures, in reality TV shows, and elsewhere. Voters are motivated to pass ``good'' resolutions and reject ``bad'' resolutions, and have zero utility if the resolution does not pass. The payoff structure is different from sequential fundraising in that {\em all} participants gain or lose when the resolution is passed, unlike in fundraising, where only investors have a payoff. This distinction disappears in unanimity games since there are either no dissenters or no utility. That is, there is no difference between sequential fundraising and sequential voting when unanimity is required. Our analysis focuses on sequential fundraising, with the somewhat more complex analysis of sequential voting handled comparatively (Section \ref{voting}).

The novelty of the mutual insurance and delegation features validly raises the question of whether we have chosen the simplest possible mode to illustrate them. We believe so. In particular, in Section \ref{non-uniform} we explain why we do not adopt the simplifying assumption of uniform qualities successfully made by \cite{Banerjee1992} and \cite{bikhchandani1992theory}. Briefly, with our forward-looking agents, an implausibly-precise knowledge of future agents' information leads to behavior that, chaotically, varies with slight variations in prior probabilities, and so has no economic significance.

%

We adopt the usual terminology of information cascades, and Bayesian games: A player's {\em type} encompasses all her private information. The state of the world, marked $\omega$, is binary, with $1$ indicating ``good'' and $0$ ``bad''. The {\em likelihood} of the state of the world $\omega$ is the ratio $\frac{\Pr[\omega = 1]}{\Pr[\omega = 0]}$. {\em Herding} occurs when a player's optimal action does not depend on her type. A {\em cascade} occurs when the same herding behavior persists for all future players. An {\em up-cascade} is when the herding is to the $a=1$ action, while a {\em down-cascade} is when it is to the $a=0$ action. The type model has the {\em monotone likelihood ratio property (MLRP)} (see, e.g., \cite{Herrera2013}, \cite{Smith2012}), when inferences from a player's type on the state of the world are monotonically increasing with the type.

A player's signal is binary, $0$ signalling a bad state of the world, and $1$, a good state of the world. The quality of the signal is the {\em a priori} probability that the signal matches $\omega$, the state of the world. It ranges from $\frac{1}{2}$, when the signal is uninformative about $\omega$, to $1$, when it unambiguously reveals it. The signal quality is privately known and individual to a player, and is independently drawn from a known quality distribution whose support is in $[\frac{1}{2},1]$.

As we will show, we can roll together, for each player, the private quality of a signal in $[\frac{1}{2},1]$ and the private binary signal of $\omega$, into a new type in $[0, 1]$ that obeys MLRP. We will also demonstrate that the players' optimal behavior is crucially affected by the {\em range} of the quality distribution. We mark by $Q$ the maximal value in the quality distribution support, and by $R$ the minimal value. $Q = 1$ when there is non-zero probability for a type whose signal deterministically reveals the state of the world. This translates to an MLRP model with unbounded support, which, as usual in unbounded models (see \cite{Smith2012}), precludes herding (because a sufficiently-informed player's signal can overcome any prior, no matter how good or bad).

We assume $Q < 1$, which leads to an MLRP model with bounded support, making herding and cascades possible. This is because the existence of players who have perfect knowledge of the state of the world appears unrealistic to us. Most questions, and especially predictions (of, e.g., the weather, sporting events, or the success of ventures), have too many unknowns to be answered exactly. This is borne out by empirical studies (e.g., \cite{goel2010prediction}, \cite{ban2018all}) that indicate that predictions are limited in their accuracy, even if formed by aggregation of several opinions.

One result emerging from our analysis is that the minimal $R$ is also of significance. $R = \frac{1}{2}$ attaches non-zero probability to a type who is perfectly ignorant of the state of the world. While not impossible, it is often implausible that a player with skin in the game is perfectly ignorant. Setting $R > \frac{1}{2}$ sets bounds on that ignorance. This bound also means that every player can trust future players to have a minimally-informative signal. With uniformly-distributed qualities, we show that players delegate only when $R$ exceeds a bound, and never delegate for $R = \frac{1}{2}$, i.e., when future players may be perfectly ignorant. If the lower bound $R$ and the higher bound $Q$ are not too close, i.e., if signal qualities are far from being uniform, only early players delegate.


In this paper we develop a model for the general problem, and formulate equations that are satisfied in equilibrium, then demonstrate interesting economic and game-theoretical aspects of the solution. We seek a Markov Perfect Equilibrium (MPE) of our problem, as defined by \cite{maskin2001markov}, where strategies are memoryless and depend only on a payoff-relevant state.  We furthermore adopt a refinement that makes the equilibrium unique, and show that the optimal strategy is pure, unique, and a threshold one, in which a player invests only if her type is at least equal to a threshold, with indifference to action possible only at the threshold. We demonstrate the cascade limits: Like in the information-cascades setting, an up-cascade occurs at public likelihoods above $\frac{Q}{1 - Q}$. Unlike the information-cascades setting, a down-cascade occurs at public likelihoods below $\big(\frac{1 - Q}{Q}\big)^B.$ Therefore the game is {\em not} in a cascade for arbitrarily unfavorable likelihoods, whenever $B$ is large enough. There is  a limit on information aggregation in every sequential fundraising, which we show is public likelihood $\big(\frac{Q}{1-Q}\big)^2$.

In the general case, $B \leq n$, except at cascades, solving for equilibrium analytically is mostly beyond our reach. However, unanimity games $B = n$, which include the Mom and Dad ($B=n=2$) problem, are generally solvable. Interesting aspects of the equilibrium include

\begin{itemize}
\item Depending on the public likelihood, players often, but not always, play the same threshold strategy.  In a sense, they divide the ``burden'' of vindicating the project equally among themselves. We show how this is related to a simultaneous version of the game, and show that players in the sequential game can play unequal thresholds only when the simultaneous game has more than one equilibrium. We provide a necessary condition for two players to play different equilibrium strategies.
\item For some quality distributions, players never delegate. We give a criterion for this.
\item For some quality distributions that allow players to delegate, only the earliest players delegate. In a game with such a quality distribution, often all early players delegate. This serial delegation, a``reverse'' cascade of sorts, ends when only a few more investments are needed from the last few players. In effect, the decision is delegated to these last few players.
\end{itemize}

In addition, we show that, in unanimity games, players always assume mutual insurance, i.e., shade their threshold lower, relying on the strategies of future players. We conjecture that this is true for non-unanimity games ($B < n$) too.



\subsection{Related Work}\label{sec:rel_lit}

The information cascades literature originated from the canonical work of \cite{Banerjee1992} and \cite{bikhchandani1992theory}. This line of literature is devoted to the study of information aggregation. In this group of models, partly informed agents act sequentially and observe the actions taken prior to their arrival. Each agent must decide on an alternative of unknown quality, using the information embedded in the observed history, combined with her own private signal. \cite{Banerjee1992} and \cite{bikhchandani1992theory} found that when agent signals are binary, eventually agents will ignore their own information and herd on the previous action taken (i.e. herding will occur). At this point information aggregation halts as agent actions are no longer informative. Furthermore, there is a positive probability that the herd was formed over the ex-post inferior alternative.  In a follow-up work, \cite{Smith2012} studied a similar model, but with a general signal structure, and showed that learning can occur if and only if signals are unbounded in their quality. Subsequent work studied the robustness of these results on various incentive schemes \cite{Mueller-frank2012,Arieli2018,Moscarini2001a}, various observational structures \cite{Acemoglu2011,Herrera2013}, and even cases where agents are affected by congestion \cite{Eyster2013}. One common feature of these, and the vast majority of existing work, is that the utility of an agent depends solely on actions taken in the past and the information available in the present. We extend this to where agents' utility is affected by the action of future arrivals.

Two rare examples of the introduction of forward-looking agents into an information cascade framework are in \cite{Arieli2018a,Smith2017exp}.  In \cite{Arieli2018a}, the agents choose between strategically priced alternatives. However, even in \cite{Arieli2018a},  the farsighted agents are the firms, while the consumers who potentially herd are myopic. In \cite{Smith2017exp}, a social planner assigns actions to agents in order to maximize the discounted sum of agent utilities.  To the best of our knowledge, we are first to successfully characterize an equilibrium in the setting of information cascades in which the agents themselves are forward-looking.

The theory of incentives in sequential fundraising processes has been gaining traction with the rise of crowdfunding platforms and venture capital funds. The literature's lion share focuses on the interaction with the entrepreneur (e.g., \cite{Hellmann2015,Bergemann2004}), studying questions of moral hazard, number of investments, and the pace at which investments occur (e.g., \cite{Bergemann2008,Halac2020}. 
  {\cite{Ali2012} presented a model in which agents arrive sequentially and have utilities that depend on both the underlying state of nature and the realized vector of actions. They defined a set of `` sincere"  strategies in which an agent's action is determined solely by the history of actions and her signal and not by her anticipation of future actions and proved that, under some assumptions, a sincere strategy equilibrium exists. Additionally, they show that herding may still emerge.  
The work closest to ours is \cite{Halac2020}. The model presented in \cite{Halac2020} comprises agents who have heterogeneous investment size, act sequentially and decide whether to invest in a firm. The firm can alter its revenue stream, which determines the investment attractiveness. They show that the optimal firm strategy increases the inequality between the returns of large investors and smaller ones.   We supplement \cite{Bergemann2008,Bergemann2004} by requiring investments to initiate the process, thus investors cannot alter their decision based on interim output. We extend \cite{Hellmann2015} to a case of more than two rounds of investment.  Compared to \cite{Ali2012}, We study a broader set of Markovian 
strategies where agents consider their influence on future agents as well as their pivotality. We prove that herding may still emerge. Furthermore, we show that information aggregation may suffer from a new type of equilibrium in which early agents {\em de facto} delegate their decision to later arrivals.} Finally, we extend \cite{Halac2020}, as in our model, agents are aware of the exact amount of participation required, and thus can make inferences from the decision of the pivotal investor, who, in turn, knows she is pivotal. We show that the pivotal nature of the latter has a major impact on her decision and on the efficiency of information aggregation, and allows for both mutual insurance and delegation which were unobserved in previous work.

Another related line of work is the literature that studies crowdfunding \cite{Alaei2016,Strausz2017,arieli2018one, Mollick2014}. In \cite{Strausz2017}, the author studied the robustness of crowdfunding campaigns to moral hazard and charaterized the conditions under which crowdfunding campaigns are resilient to entrepreneurial abuse.   Alaei et. al. \cite{Alaei2016} presented a model where agents' actions depend on past actions as well as their estimation of what  future agents will do.  Their setting, motivated by crowdfunding campaigns, has agents with private valuations, and the public information is solely the number of adoptions. We present a model where agents have a noisy signal on quality. 
In \cite{Alaei2016}'s model, the risk facing consumers is to invest in a failed campaign due to contribution costs. In ours the risk is due to a poor product. Due to this, our model can be applied to other cases, such as sequential voting and group buy. \cite{arieli2018one} presented a static model of crowdfunding, in which agent signals are binary, and show that agent participation is affected by the decision of their colleagues, when the population size is sufficiently large. They use the term ``social insurance'', roughly in the meaning meant by ``mutual insurance'' in this paper.  

In addition, we contribute to the study of sequential voting. Our uncovered delegation phenomena is simlar to the Bandwagon effect described in \cite{Callander2007}, with the only difference that agent utility in our model is only effected if she, and a sufficient proportion of her peers, choose to opt-in, while in \cite{Callander2007}, agents are forced to choose between two alternatives and gain extra utility from choosing the winner. Finally, when discussing unanimity, our results consolidate with \cite{Dekel2000}. However, we supply a rigorous treatment for the general sequential voting problem.  

The paper is organized as follows.  Our model is presented in Section \ref{model}. Section \ref{analysis} analyzes and solves the problem up to an equilibrium characterization. In Section \ref{results} we demonstrate properties of the solution in cases of interest and in the general case. Section \ref{voting} we analyze sequential voting, in comparison to our analysis of sequential fundraising. We offer conclusions and general remarks in Section \ref{discussion}.
\section{Model}
\label{model}

$n_0$ players, labeled $1, 2, \ldots, n_0$, need to take a joint decision on a project. Each has two possible actions, $1$ ``invest'' and $0$ ``decline''.
Player $1$ makes a public decision first, followed by player $2, 3 \dots$. Player $i$'s action is marked $a_i$. The project takes place only if at least $B_0$ players invest, i.e., if $a_1 + \ldots + a_{n_0} \geq B_0$. 

The state of the world, $\omega$, is either $1$ (project is ``good''), or $0$ (project is``bad'').

If the project does not take place, all player utilities are zero. If the project does take place, each player who invested gets utility $2\omega - 1$, while players who declined have zero utility.\footnote{A rescaling of likelihoods will handle the case where upside and downside are not equal.}

Each player receives a private, noisy signal $s_i \in \{0, 1\}$ of $\omega$. Signals are mutually independent, conditionally on $\omega$. The probability that the signal is correct ($s_i = \omega$) is individual to the player, called the signal {\em quality} and marked $q_i$. With probability $1 - q_i$, the signal is wrong ($s_i = 1 - \omega$). Qualities are independent, identically-distributed samples from a distribution $\bm{q} \in \Delta([R, Q])$ where $\frac{1}{2} \leq R < Q < 1$. The distribution has piecewise-continuous density $f_{\bm{q}}(x)$ with no atoms, and $f_{\bm{q}}(R + \epsilon) > 0$, $f_{\bm{q}}(Q - \epsilon) > 0$ for every sufficiently small $\epsilon > 0$. Each player privately knows her quality.

The {\em prior likelihood} of the state of the world, marked $L_0$, is publicly known. It is defined $$L_0 := \frac{\Pr[\omega = 1]}{\Pr[\omega = 0]}$$

A Markov perfect equilibrium (MPE) of the game is sought. This means a perfect Bayesian equilibrium (PBE) in {\em Markovian strategies}, as defined by \cite{maskin2001markov}. 
The strategy of each player $$S: \mathbb{R}_{>0} \times \mathbb{N} \times \mathbb{N} \times [1-Q, 1 - R] \cup [R, Q] \mapsto \Delta(\{0, 1\})$$ is accordingly a probability distribution over the action space $\{0, 1\}$, which is a function of the payoff-relevant state $(L, B, n)$, and of the player's privately-known type $t$, defined below, the state comprising
\begin{itemize}
\item $L$: The current public likelihood of $\omega$.
\item $B$: The number of investments still needed for completion.
\item $n$: The number of players still to make a decision. Note that this encodes the player's identity $i: = n_0 - n + 1$.
\end{itemize}

In accordance with the Markov property, each state $(L, B, n)$ may be considered the starting point of the fundraising, in which the initial state $(L_0, B_0, n_0)$ is irrelevant.

A player's type consists of her private information, i.e., the pair $(q_i, s_i)$. For each player $i$, we roll this pair into a single real number $t_i$, defined
\begin{align}
\label{type}
t_i = \left\{  \begin{array}{ll}
q_i & s_i = 1\\
1 - q_i & s_i = 0
\end{array} \right.
\end{align}

\noindent so that $\Pr[\omega = 1 | t_i] = t_i$. Note that $t_i$'s support is $[1 - Q, 1 - R] \cup [R, Q]$.

In addition, we make the following tie-breaking refinement:\footnote{Despite some similarities, our refinement is {\em not} Selten's Trembling-Hand Perfection. Similar tie-breaking mechanisms have been used in the context of dynamic analysis. See, for example, the {\em Obfuscation Principle} presented in \cite{Ely2017}. There, a social planner declares agent strategies before the game commences.} When a player has several strategies that are tied for having optimal utility expectation, the player {\em must} play a mixed strategy in which at least two of the tied strategies are played with a probability not less than $\epsilon$, a fixed parameter, that is arbitrarily small but is greater than zero.

\section{Analysis}
\label{analysis}

\subsection{Preliminaries}
\label{prelim}

The {\em history} of player $i$, $h_{i-1} \in \{0, 1\}^{i-1}$ is the list of actions of previous players $h_{i-1} := (a_1, \ldots, a_{i-1})$ ($h_0$ is an empty list).

The {\em type-independent} expectation of a player's random variable is the expectation calculated by an observer who has no access to the player's privately-known type, and so may be different from the player's own calculated expectation for the same variable.

Mark by $L_i$ the public likelihood of the state of the world after player $i$'s action, inferred from $L_0$ and $h_i$, i.e. 
\begin{align}
\label{L_i}
L_i := \frac{\Pr[\omega = 1 | h_i]}{\Pr[\omega = 0 | h_i]}
\end{align}

The private likelihood of a player, given her privately-known type (as the type and the history are mutually independent), is
\begin{align}
\label{private-likelihood}
\frac{\Pr[\omega = 1 | h_{i-1}, t_i]}{\Pr[\omega = 0 | h_{i-1}, t_i]} = \frac{\Pr[\omega = 1 | h_{i-1}]}{\Pr[\omega = 0 | h_{i-1}]} \frac{\Pr[\omega = 1 | t_i]}{\Pr[\omega = 0 | t_i]} = L_{i-1} \frac{t_i}{1 - t_1}
\end{align}

We will show (Corollary \ref{c}) that, except for {\em irregular states}, defined below, there is a unique MPE, and equilibrium strategies are pure threshold strategies, i.e., players invest if their type is larger than the threshold, decline if below it, and are possibly indifferent if their type is equal to the threshold. The strategy is therefore represented by a function\footnote{The number of players remaining is part of the state, so no two players can be in the same state. Thus, we can use a common strategy function $\sigma$.} $$\sigma: \mathbb{R}_{>0} \times \mathbb{N} \times \mathbb{N} \mapsto (0, 1)$$ where $S(L, B, n, t) = 1$ if $\sigma(L, B, n) < t$ and $S(L, B, n, t) = 0$ if $\sigma(L, B, n) > t$.

\subsection{The Case $B \leq 1$: Information Cascades}
When $B \leq 1$ we are in a situation similar to that described in the classical model of \cite{Banerjee1992,bikhchandani1992theory}: The project has already been decided on, or will be if the current player invests. Therefore the behavior of future players has no bearing on the current player's strategy: There is no mutual insurance, and the player decides strictly based on her private likelihood \eqref{private-likelihood}, investing when it is $> 1$ and declining when it is $< 1$.

\begin{proposition}[$B \leq 1$]
\label{B=1}
Given $B \leq 1$
\begin{align}
\label{last-strategy} 
\sigma(L, B, n) &= \frac{1}{1 + L}
\end{align}
\end{proposition}

\begin{proof}
If the number of investments needed is reached or exceeded, the probability for completion is $1$. A player with type $t$ has expectation for investing 
\begin{align*}
U(L, B, n, t) &= \Pr[\omega = 1 | L, t]  - \Pr[\omega = 0 | L, t] = \Big(1 - \frac{1}{1 + L \frac{t}{1-t}}\Big) - \frac{1}{1 + L \frac{t}{1-t}}
\end{align*}
\noindent which is non-negative iff $L \frac{t}{1-t} \geq 1$, i.e., when $t \geq \frac{1}{1 + L}$.
\end{proof}

\subsection{Probability of Completion}
\label{recur}

We will show that a player's optimal strategy critically depends on the {\em probability of completion}, which we now define.

In an $n$-player game, let $\bm{\phi} = (\phi_1, \ldots, \phi_n)$ be a profile of strategies in MPE for each player.

Define the {\em probability of completion}, $\pi_\omega^{\bm{\phi}}(L, B, n)$ as the probability that the project takes place (i.e., the fundraising completes) when the state of the world is $\omega \in \{0, 1\}$ and the players play their strategies in $\bm{\phi}$, given public likelihood $L$, $B$ outstanding investments and $n$ remaining players. 
Clearly, $\pi_\omega^{\bm{\phi}}(L, B, n) = 0$ when $B > n$ and $\pi_\omega^{\bm{\phi}}(L, B, n) = 1$ when $B \leq 0$,  for any $L, \omega$ and $\bm{\phi}$.

These are examples of states where the probability of completion does not depend on the strategy profile. In general, if there is a unique MPE, there is a unique probability of completion, which we mark simply $\pi_\omega(L,B,n)$. We shall later derive a recurrence equation for it in Corollary \ref{recurrence}.

States $(L, B, n)$ where there are multiple MPE's, or where the strategy $\sigma(L, B, n)$ is not continuous in $L$, are called {\em irregular}, and are excluded from our discussion. We later (Section \ref{regular}) characterize when a state is irregular, and show why it is reasonable to not consider them.

%
%
%

Let $\phi_1$ be some Markovian strategy that the first player plays in state $(L, B, n)$. It depends on the player's private type. Other players, not knowing the first player's type, make inferences based on their knowledge of $\phi_1$. Let the type-independent probability (the probability calculated by an observer who has no private information of the player's type) that the player declines under $\phi_1$ and state of the world $\omega$ be marked $F_\omega(\phi_1)$:\footnote{In a slight abuse of notation, we shall later (Section \ref{inferences}) intentionally reuse $F(\cdot)$ as the c.d.f. of the type distribution, once we have proven that the only strategies that need be considered are threshold strategies.} 
By Bayes' rule, if the player invests, the posterior likelihood changes to
\begin{align}
\label{L+}
L^+ := L \frac{1 - F_1(\phi_1)}{1 - F_0(\phi_1)}
\end{align} \noindent and the state changes to $(L^+, B-1, n-1)$, while if the player declines, the posterior likelihood changes, by Bayes' rule, to 
\begin{align}
\label{L-}
L^- := L \frac{F_1(\phi_1)}{F_0(\phi_1)}
\end{align} \noindent and the state changes to $(L^-, B, n-1)$.


If neither of these two states are irregular, there is a well-defined probability of completion for every possible action of the player. Thus, the probability of completion under $\phi_1$, which we mark $\pi_\omega^{\phi_1}(L, B, n)$, is well-defined, equalling
\begin{align}
\label{sigma1}
\pi_\omega^{\phi_1}(L, B, n) = [1 - F_\omega(\phi_1)] \pi_\omega\big(L \frac{1 - F_1(\phi_1)}{1 - F_0(\phi_1)}, B - 1, n - 1\big) + F_\omega(\phi_1) \pi_\omega\big(L\frac{F_1(\phi_1)}{F_0(\phi_1)}, B, n - 1\big).
\end{align}

\subsection{Equilibrium Strategies are Threshold Strategies}

We now prove that the only strategies that are played in MPE are threshold strategies.


For $x \in (0, 1)$, we define $\mathcal{S}(x)$ as a strategy with threshold $x$,  investing for every type $> x$, declining for every type $< x$, and possibly indifferent for type $= x$.

\begin{theorem} [Threshold Strategy]\label{thm:threshold}
\label{threshold-strategy}
Excluding irregular states $(L, B, n)$, let $\bm{\phi} := (\phi_1, \ldots, \phi_n)$ be an MPE with public likelihood $L$, $B$ outstanding investments, and $n$ remaining players. Then $\phi_1$ is a threshold strategy $\mathcal{S}(x)$, where $x \in (0,1)$.
\end{theorem}

\begin{proof}
Let a player have a type $t$ with public likelihood $L$, $B$ outstanding investments and $n$ remaining players. Her private posterior likelihood $L'$ is, by \eqref{private-likelihood}, $L \frac{t}{1-t}$, and her probability for $\omega = 0$ is $\frac{1}{1 + L'} = \frac{1}{1 + L \frac{t}{1-t}}$.

Let the type-independent probability for declining under $\phi_1$ be $F_\omega(\phi_1)$. The posterior likelihood from investing, $L^+$, is given in \eqref{L+}. $L^+$ is independent of $t$, since the type is private and unobserved. So, If the player invests, her expectation, marked $U^{\phi_1}(L, B, n, t)$, is 
\begin{align}
U^{\phi_1}(L, B, n, t) &= \Pr[\omega = 1 | L, t] \pi_1(L^+, B - 1, n - 1) - \Pr[\omega = 0 | L, t] \pi_0(L^+, B - 1, n - 1) \nonumber \\
\label{expectation-phi}
&= \Big(1 - \frac{1}{1 + L \frac{t}{1-t}}\Big) \pi_1(L^+, B - 1, n - 1) - \frac{1}{1 + L \frac{t}{1-t}} \pi_0(L^+, B - 1, n- 1).
\end{align}

$U^{\phi_1}(L, B, n, t)$ is non-negative iff
\begin{align}
\label{invest}
L \frac{t}{1-t} \geq \frac{\pi_0(L^+, B - 1, n - 1)}{\pi_1(L^+, B - 1, n- 1)}.
\end{align}

The right-hand side of \eqref{invest} does not depend on type $t$, while the left-hand side increases with increasing $t$, and varies continuously from $0$ to $\infty$ when $t$ varies from $0$ to $1$. It follows
\begin{enumerate}
\item By the mean-value theorem, there exists a $t = x$ where \eqref{invest} holds with equality.
\item If $\phi_1$ calls for declining at any $t > x$, the player will deviate, as $U^{\phi_1}(L, B, n, t) > 0$, and so the strategy is not playable in equilibrium.
\item If $\phi_1$ calls for investing at any $t < x$, the player will deviate, as $U^{\phi_1}(L, B, n, t) < 0$, and so the strategy is not playable in equilibrium.
\end{enumerate}
Therefore the only strategies playable in equilibrium are threshold strategies $\mathcal{S}(x)$ where for $t = x$ \eqref{invest} holds with equality.
\end{proof}

\subsection{Inferences from Actions in Threshold Strategies}
\label{inferences}
At equilibrium, by Theorem \ref{thm:threshold}, all players play threshold strategies, which are commonly known. A player's type is not observable, but the fact that she invested proves that her type is 
greater than her threshold, while if she declined, the inference is that her type is below the threshold. Since these events have, {\em a priori}, different conditional probabilities for each state of the world $\omega$, observing the action leads, by Bayes' rule, to an updated public likelihood of $\omega$.

We derive from $f_{\bm{q}}(x)$, the quality density function, the distribution of type $t$. For each state of the world $\omega$, $t$ is a random variable with support in $[1 - Q, 1 - R] \cup [R, Q]$, with density $f_\omega(y)$ and c.d.f. $F_\omega(y)$. By \eqref{type} and the definition of the quality $q$ 
\begin{align}
\label{pdf-1}
f_1(y) &= \left\{  \begin{array}{ll}
y f_{\bm{q}}(y) & y \in [\frac{1}{2}, Q] \\
y f_{\bm{q}}(1 - y) & y \in [1 - Q, \frac{1}{2}]
\end{array} \right. \\
\label{pdf-0}
f_0(y) &= \left\{  \begin{array}{ll}
(1 - y) f_{\bm{q}}(y) & y \in [\frac{1}{2}, Q] \\
(1 - y) f_{\bm{q}}(1 - y) & y \in [1 - Q, \frac{1}{2}]
\end{array} \right.
\end{align}

Thus
\begin{align}
\label{mlrp}
\frac{f_1(y)}{f_0(y)} = \frac{y}{1-y}
\end{align}
\noindent regardless of $f_{\bm{q}}(\cdot)$, and is monotonically increasing in $y$, so type distributions always have the MLRP property.

For $x \geq Q$, the undefined $\frac{1 - F_1(x)}{1 - F_0(x)} = \frac{0}{0}$ is taken to be the limit by L'H\^opital's rule
\begin{align}
\label{limQ}
\lim\limits_{x \to Q} \frac{1 - F_1(x)}{1 - F_0(x)} = \frac{Q}{1-Q}.
\end{align}

For $x \leq 1-Q$, the undefined $\frac{F_1(x)}{F_0(x)} = \frac{0}{0}$ is taken to be the limit by L'H\^opital's rule
\begin{align}
\label{lim1-Q}
\lim\limits_{x \to 1 - Q} \frac{F_1(x)}{F_0(x)} = \frac{1 - Q}{Q}.
\end{align}

The following lemma will be useful.

\begin{lemma}
\label{monotone}
$\frac{1 - F_1(x)}{1 - F_0(x)}$ and $\frac{F_1(x)}{F_0(x)}$ are monotonically increasing in $x$.
\end{lemma}

\begin{proof}
See Appendix.
\end{proof}

Since $F_1(1) = F_0(1) = 1$ we deduce from Lemma \ref{monotone} first-order stochastic dominance
\begin{align}
\label{dominance}
F_0(x) \geq F_1(x)
\end{align}
\noindent with equality only when $x \notin (1-Q, Q)$.

When the public likelihood is $L$, the probability for $\omega = 1$  is $\frac{L}{1 + L}$. Thus, observing an investment ($a_i = 1$) by a player whose threshold strategy is $x$, when the prior public likelihood is $L$, we derive by Bayes' rule
\begin{align}
\label{for_a=1}
\Pr[\omega = 1 | a_i = 1] &= \frac{\frac{L}{1 + L} \Pr[a_i = 1 | \omega = 1]}{\frac{L}{1 + L} \Pr[a_i = 1 | \omega = 1] + \frac{1}{1 + L} \Pr[a_i = 1 | \omega = 0]} \nonumber \\
&= \frac{1}{1 + \frac{1}{L} \frac{ \Pr[a_i = 1 | \omega = 0]}{\Pr[a_i = 1 | \omega = 1]}} = \frac{1}{1 + \frac{1}{L} \frac{1 - F_0(x)}{1 - F_1(x)}}.
\end{align}

From which we conclude that the posterior public likelihood $L^+(L, x)$ inferred from investment is
\begin{align}
\label{infer-invest}
L^+(L, x) := \frac{\Pr[\omega = 1 | a_i = 1]}{\Pr[\omega = 0 | a_i = 1]} = L \frac{1 - F_1(x)}{1 - F_0(x)}.
\end{align}

Similarly, the posterior public likelihood $L^-(L, x)$ inferred from a decline is
\begin{align}
\label{infer-decline}
L^-(L, x) :=  \frac{\Pr[\omega = 1 | a_i = 0]}{\Pr[\omega = 0 | a_i = 0]} = L \frac{F_1(x)}{F_0(x)}.
\end{align}

Due to stochastic dominance \eqref{dominance}, we have, for every $L, x$ $$L^-(L, x) \leq L \leq L^+(L, x).$$

Using the above, Theorem \ref{threshold-strategy}, and \eqref{invest} in its proof, we derive a condition which is fulfilled by every threshold strategy played in MPE.
\begin{corollary}[Threshold Indifference Condition]
\label{indifference}
In equilibrium, a threshold $x$ strategy satisfies
\begin{align}
\label{threshold}
L \frac{x}{1-x} = \frac{\pi_0\big(L \frac{1 - F_1(x)}{1 - F_0(x)}, B - 1, n - 1\big)}{\pi_1\big(L \frac{1 - F_1(x)}{1 - F_0(x)}, B - 1, n- 1\big)}.
\end{align}
\end{corollary}

We now derive a recurrence equation for the probability of completion $\pi_\omega$.

\begin{corollary}[Probability of Completion]
\label{recurrence}
Except at irregular states $(L, B, n)$, the probability of completion is continuous in $L$, equalling
\s\begin{align}
\label{pi-recurrence}
\pi_\omega(L, B, n) = \left\{  \begin{array}{ll}
1 & B \leq 0 \\
0 & n \leq 0 \land B > 0 \\
 \big[ 1 - F_\omega(x)\big]\pi_\omega\big(L \frac{1 - F_1(x)}{1 - F_0(x)}, B - 1, n - 1\big) + F_\omega(x) \pi_\omega\big(L\frac{F_1(x)}{F_0(x)}, B, n - 1\big) & $otherwise$
\end{array} \right.
\end{align}\n
\noindent where $x = \sigma(L, B, n)$.
\end{corollary}

\begin{proof}
This follows from \eqref{sigma1} for a threshold strategy $\phi_1 = \mathcal{S}(x)$, that is played in a unique MPE. $F_\omega(x)$ is continuous in $x$, and therefore, since the state is regular, continuous in $L$. So by induction on $n$, using \eqref{pi-recurrence}, $\pi_\omega$ is continuous in $L$, wherever it is well-defined.
\end{proof}

\subsection{The Equilibrium Strategy}
\label{optimal}

In previous sections we show that, in MPE, a player plays a threshold strategy where the threshold satisfies \eqref{threshold}. If \eqref{threshold} has a single solution $x$, this single solution is the player's strategy in MPE. But, if there are multiple $x$'s satisfying \eqref{threshold}, she must choose.

\begin{proposition}[Optimal Strategy]
\label{maximal}
In MPE, a player's strategy is a Markovian strategy whose type-independent utility expectation is not exceeded by any other Markovian strategy.
\end{proposition}

\begin{proof}
In games in general, in MPE a player plays a strategy that is a function of the payoff-relevant state, and is a best response to the strategies of all concurrent and future players. The strategy does {\em not} depend on the strategies of {\em past} players, as their strategies, to the extent that they are payoff-relevant, are already encompassed in the state.

In the sequential fundraising game a player has no concurrent players. A {\em playable} strategy is a strategy from which the player will not deviate for any type. The player's optimal choice is the playable strategy that has maximal utility expectation {\em before she knows her type}, given the state and the strategies of {\em future} players. No simultaneous change of strategy by future players is implied by this strategy choice, since, for future players the current player is a {\em past} player whose strategy does not affect theirs.\footnote{This argument would not work in an SPE/PBE solution concept, where a player's strategy is a best response to all other strategies, past or present.}
\end{proof}

%
\noindent The type-independent utility expectation is given by the following proposition.

\begin{proposition}[Utility Expectation]
\label{expectation}
Excluding irregular states $(L, B, n)$, the type-independent utility expectation $U^{\mathcal{S}(x)}(L, B, n)$ for a strategy $\mathcal{S}(x)$ is
\begin{align}
\label{pi_diff}
U^{\mathcal{S}(x)}(L, B, n) = \frac{L}{1 + L}&[1 - F_1(x)] \pi_1\big(L \frac{1 - F_1(x)}{1 - F_0(x)}, B - 1, n - 1\big) \nonumber \\
&-  \frac{1}{1 + L}[1 - F_0(x)] \pi_0\big(L \frac{1 - F_1(x)}{1 - F_0(x)}, B - 1, n - 1\big).
\end{align}
\end{proposition}

\begin{proof}
The expectation for declining is $0$, so $U^{\mathcal{S}(x)}(L, B, n)$ is the expectation for the remaining action, investing. The probability that the player invests is $1 - F_\omega(x)$, and so, when $\omega=1$, which occurs with probability $\frac{L}{1+L}$, the player gains $1$ with probability $$[1 - F_1(x)] \pi_1\big(L \frac{1 - F_1(x)}{1 - F_0(x)}, B - 1, n - 1\big)$$ Similarly, when $\omega=0$,  which occurs with probability $\frac{1}{1+L}$, the player loses $1$ with probability $$[1 - F_0(x)] \pi_0\big(L \frac{1 - F_1(x)}{1 - F_0(x)}, B - 1, n - 1\big)$$
\end{proof}

\subsection{Equally-Optimal Strategies}
\label{equal}

When a player has several strategies with maximal type-independent utility expectation, we show that the player nevertheless prefers one of these strategies, when we assume, according to our model's refinement, that the player always mixes between these strategies.

\begin{proposition}[Tie-Breaking Equally-Optimal Strategies]
\label{equally-optimal}
Excluding at irregular states $(L, B, n)$, assume two or more strategies $\mathcal{S}(x_1), \ldots, \mathcal{S}(x_k)$ have equal type-independent utility expectation. A player who mixes between the strategies, assigning probability at least $\epsilon > 0$ to at least two of them, maximizes her utility expectation by assigning maximal probability $1 - \epsilon$ to strategy $\mathcal{S}(x_i)$ that maximizes the {\em discriminator}
\s\begin{align}
\label{discriminator}
D(x_i) &:=  \max_{j \in [n], j \neq i} \Big\{L\pi_1\big(L \frac{1 - F_1(x_i)}{1 - F_0(x_i)}, B - 1, n - 1\big) [1 - F_1(x_j)] -  \pi_0\big(L \frac{1 - F_1(x_i)}{1 - F_0(x_i)}, B - 1, n - 1\big) [1 - F_0(x_j)]\Big\}
\end{align}\n
\end{proposition}

\begin{proof}
See Appendix.
\end{proof}

\subsection{Equilibrium Characterization and Irregular States}
\label{regular}

We summarize the above sections with a characterization of the MPE of a sequential fundraising.

\begin{corollary}[Equilibrium Characterization]
\label{c}
At every state $(L, B, n)$ that is not irregular, there is a unique pure MPE in which, in every subgame, all players play threshold strategies. The playable threshold strategies $\mathcal{S}(x)$ are those that satisfy \eqref{threshold}, and the strategy played in MPE is the one with highest type-independent utility expectation, as given in \eqref{pi_diff}. In case of a tie between strategies for maximal utility expectation, the maximal-utility strategy with highest value of discriminator \eqref{discriminator} is played in MPE. The state is irregular if \eqref{discriminator} fails to discriminate between two or more equally-optimal strategies in any subgame, or if the uniquely-optimal strategy's threshold is not continuous at $L$ for the given $B, n$.
\end{corollary}

As the corollary states, an irregular state is one where any subgame has multiple MPE's, or if an arbitrarily-small change in $L$ results in a finite change in the equilibrium threshold (as, e.g., in Figure \ref{fig:6580} at $L \simeq 2.2$). Consequently, the following is a recursive characterization of these states.

\begin{corollary}[Irregular States]
\label{irregular}
The state $(L, B, n)$ is irregular
\begin{enumerate}
\item Never, for $B \leq 1$.
\item \label{ii} If it has multiple MPE's, with several equally-optimal strategies whose discriminators \eqref{discriminator} are equal.
\item \label{iii} If it has a single MPE, and the first player plays $\mathcal{S}(x)$, where $x$ is not continuous at $L$ for the given $B, n$.
\item \label{iv} If it has an MPE where the first player plays $\mathcal{S}(x)$, and any of the states $(L \frac{1 - F_1(x)}{1 - F_0(x)}, B-1, n-1)$ or $(L \frac{F_1(x)}{F_0(x)}, B, n-1)$ is irregular.
\end{enumerate}
\end{corollary}

The irregular states ultimately depend on the quality distribution $\bm{q}$, but this characterization shows why they are most probably rare, and, though this may be hard to prove, often do not occur at all. We conjecture that their number is finite for bounded $n$, basing ourselves on the following reasoning: Every equilibrium threshold is a solution of \eqref{threshold}, an equation which is not under-specified, and so typically has a finite number of solutions, whose number is further reduced by ranking them, first, by utility expectation and second, by the discriminator \eqref{discriminator}. Corollary\ref{irregular}(\ref{ii}) and Corollary\ref{irregular}(\ref{iii}) stem from the discriminator not excluding all but one solution of \eqref{threshold}, whether because there remains more than one equally-optimal solution, or whether an infinitesimal change in $L$ changes which solution is played in equilibrium. Corollary \ref{irregular}(\ref{iv}) shows that the number of irregular states at least doubles when the number of players is increased by $1$.

\section{Solution Properties}
\label{results}

\subsection{Herding and Cascades}
\label{cascades}

Herding describes a situation where all types have the same optimal action. A cascade occurs when herding persists among players, so that when a cascade starts, all future players decline (down-cascade) or all invest (up-cascade).

%

An up-cascade occurs wherever the public likelihood is at least $\frac{Q}{1-Q}$.
\begin{theorem}[Up-Cascade]
\label{up-cascade}
For every $B, n$ and $B \leq n$, $\sigma(L, B, n) \leq 1 - Q$ and an up-cascade is in progress iff $L \geq \frac{Q}{1 - Q}$.
\end{theorem}

\begin{proof}
An up-cascade occurs iff the probability of completion is $1$. We prove that $\pi_\omega(L, B, n) = 1$ when $L \geq \frac{Q}{1 - Q}$ by induction on $B$. For $B = 0$ this is true for any $L$. Assume the theorem true up to $B-1$. Then, a threshold strategy $x$ must satisfy, by \eqref{threshold} and the induction hypothesis $L \frac{x}{1-x} = 1$. So
$$\frac{Q}{1 - Q} \leq L = \frac{1-x}{x}$$ Thus $x \leq 1-Q$.
\end{proof}

At the opposite end, $\sigma(L, B, n) \geq Q$ for a herd of decliners. Since the expectation of such a strategy is $0$, it can be optimal only if it is unique. Thus from \eqref{threshold}
\begin{corollary}[Herd Decline]
\label{herd-decline}$\sigma(L, B, n) \geq Q$, i.e., declining is optimal for all types iff
\begin{align*}
L \frac{Q}{1 - Q} \leq \frac{\pi_0(L \frac{Q}{1 - Q}, B - 1, n - 1)}{\pi_1(L\frac{Q}{1 - Q}, B - 1, n - 1)}
\end{align*}
\end{corollary}

In a down-cascade, all remaining players decline. Using Corollary \ref{herd-decline} repeatedly we get
\begin{corollary}[Down-Cascade]
\label{down-cascade}
For every $B, n$ and $B \leq n$, $\sigma(L, B, n) \geq Q$ and a down-cascade is in progress iff $L \leq \Big(\frac{1 - Q}{Q}\Big)^B$.
\end{corollary}

For herding on investments, a necessary condition is \eqref{threshold} with $x \leq 1 - Q$. It is not sufficient, as it may not have maximal utility expectation.
%
\begin{corollary}[Herd Invest]
\label{herd-invest}
If $\sigma(L, B,  n) \leq 1 - Q$ then
\begin{align*}
L \frac{1 - Q}{Q} \geq \frac{\pi_0(L, B - 1, n - 1)}{\pi_1(L, B - 1, n - 1)}
\end{align*}
\end{corollary}

\subsection{Information Aggregation}

Information aggregation, or learning, occurs when the state of the world $\omega$ becomes known with higher probability, or likelihood. So we ask, what is the public likelihood $L_{end} := L_n$, at  completion of the fundraising?

To bound $L_{end}$, we observe that information aggregation stops in an up-cascade. Therefore, the public likelihood before the last player's action $L_{n-1}$ is bounded above by $\frac{Q}{1-Q}$, the up-cascade limit by Theorem \ref{up-cascade}. The inference from the last player's investment is bounded by $\frac{Q}{1-Q}$, by \eqref{limQ} and Lemma \ref{monotone}. It follows

\begin{corollary}[Learning Bound]
\label{Lend}
Let $B=n$. In equilibrium, $L_{end} < \Big(\frac{Q}{1-Q}\Big)^2$.
\end{corollary}

Furthermore, if the quality distribution is such that cascades cannot start (see below, the condition of Theorem \ref{delegation}(\ref{can-delegate})), then the bound is even lower: $L_{end} < \frac{Q}{1-Q}$.

\subsection{Example: Uniformly-Dense Qualities}
\label{uniform}

The simplest assumption we can make about the distribution of qualities is that they are uniformly distributed between their lower bound $R$ and upper bound $Q$. Here, we work out the resulting distributions and player behavior in some sequential fundraising games.

 Let qualities be uniformly dense, i.e., for $\frac{1}{2} \leq R < Q < 1$
\begin{align}
\label{pdf-u}
f_{\bm{q}}(y) &= \left\{  \begin{array}{ll}
\frac{1}{Q - R} & y \in [R, Q] \\
0 & $otherwise$
\end{array} \right.
\end{align}

Substituting \eqref{pdf-u} in \eqref{pdf-1} and \eqref{pdf-0}, we have $f_1(y) = \frac{y}{Q - R}$, and $f_0(y) = \frac{1 - y}{Q - R}$, so, for $y \in [1 - Q, Q]$
\begin{align}
\label{below-1}
F_1(y) &= \int\limits_{1-Q}^y f_1(z) dz =  \left\{  \begin{array}{ll}
 \frac{y^2 - (1-Q)^2}{2(Q - R)} & y \leq 1 - R \\
1 - \frac{Q + R}{2} & 1-R \leq y \leq R \\
1 - \frac{Q^2 - y^2}{2(Q - R)} & y \geq R
\end{array} \right. \\
\label{below-0}
F_0(y) &= \int\limits_{1-Q}^y f_0(z) dz = \left\{  \begin{array}{ll}
 \frac{Q^2 - (1-y)^2}{2(Q - R)} & y \leq 1 - R \\
\frac{Q + R}{2} & 1-R \leq y \leq R \\
1 - \frac{(1-y)^2 - (1-Q)^2}{2(Q - R)} & y \geq R
\end{array} \right.
\end{align}

\begin{figure}[tb]
\centering
\begin{minipage}{.5\textwidth}
  \centering
  \includegraphics[width=0.9\linewidth]{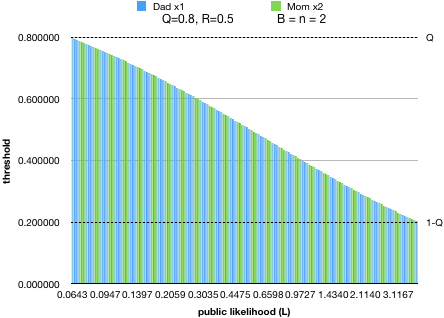}
  \caption{$n=B=2, \bm{q} \sim U(0.5,0.8)$}
  \label{fig:5080}
\end{minipage}%
\begin{minipage}{.5\textwidth}
  \centering
  \includegraphics[width=0.9\linewidth]{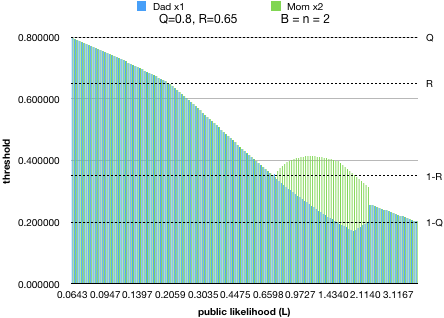}
  \caption{$n=B=2, \bm{q} \sim U(0.65,0.8)$}
  \label{fig:6580}
\end{minipage}
\centering
\begin{minipage}{.5\textwidth}
  \centering
  \includegraphics[width=0.9\linewidth]{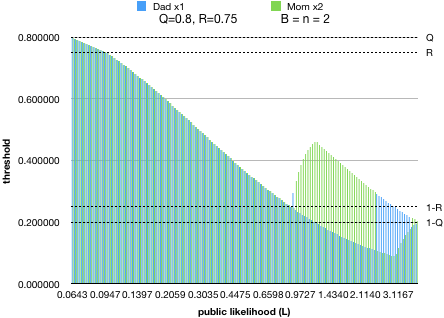}
  \caption{$n=B=2, \bm{q} \sim U(0.75,0.8)$}\n
  \label{fig:7580}
\end{minipage}%
\begin{minipage}{.5\textwidth}
  \centering
  \includegraphics[width=0.9\linewidth]{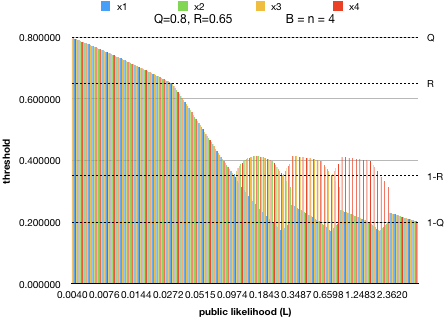}
  \caption{$n=B=4, \bm{q} \sim U(0.65,0.8)$}\n
  \label{fig:6580-4}
\end{minipage}%
\end{figure}

We use this example distribution to demonstrate several facts:
\begin{enumerate}
\item In the information cascades setting ($B \leq 1$), some quality distributions allow cascades to start, while in others, if the prior is not {\em a priori} a cascade (i.e. if $\frac{1-Q}{Q} < L < \frac{Q}{1-Q}$), a cascade never starts. We shall later (Section \ref{delegation1}) show that this property is intimately tied to delegation in sequential fundraising.

\begin{proposition}[Minimum $R$ for Cascades]
\label{min-R}
When qualities are uniformly distributed $U(R,Q)$, cascades cannot start in the information cascades setting ($B \leq 1$) when $R = \frac{1}{2}$. For each $Q$ there is a minimal $R_{min} > \frac{1}{2}$, where cascades are possible iff $R \geq R_{min}$.
\end{proposition}

\begin{proof}
See Appendix.
\end{proof}

\item We plot solutions of the Mom-and-Dad problem ($B=n=2$) for this uniformly-distributed qualities with $Q=0.8, R=0.5$ (Figure \ref{fig:5080}), $Q=0.8, R=0.65$ (Figure \ref{fig:6580}) and $Q=0.8, R=0.75$ (Figure \ref{fig:7580}). We also plot the solution for $4$ players ($B=n=4$) with $Q=0.8, R=0.65$ (Figure \ref{fig:6580-4}). Since in most of the range, the players play equal thresholds, represented as colored bars, the plots appear to have a color mixture effect.

All thresholds below $1-Q$ (herd invest, or delegate) are equivalent, as are all thresholds in the range $[1-R,R]$ (invest on $s=1$), since they fall outside the type support $[1-Q, 1-R] \cup [R, Q]$. Their different values reflect a different point of indifference in \eqref{threshold}.

In the first plot, $Q=0.8, R=0.5$, the two parents always play the same threshold strategies, and no delegation ever takes place. In the 2nd and 3rd, $Q=0.8, R=0.65$ and $Q=0.8, R=0.75$, the parents play different thresholds in part of the public likelihood range. For $Q=0.8, R=0.65$, Dad always plays a threshold equal or lower to Mom's (and so is the only one to delegate), while for $Q=0.8, R=0.75$, both parents may play the lower threshold, and both may delegate.

In Figure \ref{fig:6580-4}, the players play increasingly high thresholds from first ($x_1$) to last ($x_4$), with delegation by early players only, same as for the $B=n=2$ solution for the same quality distribution (Figure \ref{fig:6580}).

\item The threshold strategy can be non-monotonic. This is demonstrated by $Q=0.8, R=0.65$ (Figure \ref{fig:6580}), where, as the plot shows, the first player (Dad) has a non-monotonicity in his threshold at $L \simeq 2.2$.
\end{enumerate}

\subsection{Unanimity Decision}
\label{B=n}

In this section we consider equilibrium behavior in unanimity games.

\subsubsection{Equilibrium Characterization}

The equilibrium characterization given in Corollary \ref{c} takes a special form in unanimity games.

\begin{proposition} [Equilibrium Characterization]
\label{charuna}
Let $B=n$. Except at irregular states $(L, B, n)$, let $x_1, \ldots, x_n$ be the thresholds of the MPE strategies for each player. Then,

The probability of completion is
\begin{align}
\label{una-complete}
\pi_\omega(L, B, n) = \prod\limits_{k = 1}^n [1 - F_\omega(x_k)]
\end{align}

\noindent and for every $i$
\begin{align}
\label{una-threshold}
L \frac{x_i}{1 - x_i} = \prod\limits_{k \in [n], k \neq i} \frac{1 - F_0(x_k)}{1 - F_1(x_k)} 
\end{align}

The strategies $\mathcal{S}(x_1), \ldots, \mathcal{S}(x_n)$ have equally-maximal utility expectation
\begin{align}
\label{una-utility}
U^{\mathcal{S}(x_i)}(L, B, n) &= \frac{L}{1+L}  \prod\limits_{k = 1}^n [1 - F_1(x_k)] -  \frac{1}{1+L}\prod\limits_{k = 1}^n [1 - F_0(x_k)]
\end{align}

\noindent with $\mathcal{S}(x_1)$ having the largest discriminator, which, for $\mathcal{S}(x_i)$ is
\begin{align}
\label{una-discriminator}
D(x_i) = \max_{j \in [n], j \neq i} \Big\{L \frac{1 - F_1(x_j)}{1 - F_1(x_i)} \prod\limits_{k=1}^n [1 - F_1(x_k)] -  \frac{1 - F_0(x_j)}{1 - F_0(x_i)}\prod\limits_{k=1}^n [1 - F_0(x_k)]\Big\}
\end{align}
\end{proposition}

\begin{proof}
\eqref{una-complete} follows from the fact that the fundraising completes iff all players invest.

Player $i$ see public liklihood $L_{i-1} = L \prod\limits_{k \in [n], k < i} \frac{1 - F_1(x_k)}{1 - F_0(x_k)}$, and her threshold indifference \eqref{threshold} implies $L_{i-1} \frac{x_i}{1 - x_i} = \frac{\pi_0(L_i, B - i, n - i)}{\pi_1(L_i, B - i, n - i)}$, from which \eqref{una-threshold} follows.

For each $i \in [n]$, $\mathcal{L}(x_i)$, $i \in [n]$, is playable by the first player, because the threshold indifference condition \eqref{una-threshold} holds under every permutation of the thresholds. Furthermore, expectation \eqref{pi_diff} for the unanimity game is \eqref{una-utility}, due to \eqref{una-complete}, and is symmetric under permutation of the thresholds, so every $\mathcal{S}(x_i)$ has the same utility expectation. This symmetry is broken by the discriminator \eqref{discriminator}, which, in the unanimity game takes the form \eqref{una-discriminator}, due to \eqref{una-complete}.
\end{proof}

\subsubsection{Mutual Insurance}

We show that players in unanimity games always assume {\em mutual insurance}, meaning that, given the same public likelihood, a player shades her threshold lower than would be played in a single-player game (the information-cascades setting with $B \leq 1$), i.e., that for every $L, n$ $\sigma(L, n, n) \leq \sigma(L, 1, 1)$. Indeed, we get the following, stronger result,

\begin{theorem}[Mutual Insurance]
\label{mutual-insurance}
Let $B = n$, and let $\mathcal{S}(x_i)$ be the strategy in MPE of player $i \in [n]$ with prior likelihood $L$. In equilibrium $x_i \leq \sigma(L, 1, 1)$ with equality only for the last player or in a cascade.
\end{theorem}

\begin{proof}
Let the thresholds of the players in MPE be $x_1, \ldots, x_n$. By \eqref{una-threshold}
\begin{align}
\label{x_1}
L \frac{x_1}{1 - x_1} = \prod_{i = 2}^{n} \frac{1 - F_0(x_i)}{1 - F_1(x_i)}
\end{align}

From \eqref{dominance}, $\frac{1 - F_0(x)}{1 - F_1(x)} \leq 1$, which is strict except at cascades. So, from \eqref{x_1} $L \frac{x_1}{1 - x_1} \leq 1$, and consequently $x_1 \leq \frac{1}{1 + L}$, lower than the threshold played for a single player, by Proposition \ref{B=1}, and strictly so except at cascades.
\end{proof}


\subsubsection{Delegation}
\label{delegation1}

An extreme form of mutual insurance is delegation, defined as herding on investment when not in a cascade. The following theorem summarizes properties of delegation.
%

\begin{theorem}[Delegation]
\label{delegation}
Let $B = n$. Except at irregular states $(L, B, n)$, let $x_1, \ldots, x_n$ be the thresholds of the MPE strategies for each player. In equilibrium,
\begin{enumerate}
\item \label{can-delegate} In a quality distribution $\bm{q}$, delegation takes place, for some $L$, iff there exists $x \in (1 - Q, \frac{1}{2})$ for which 
\begin{align}
\label{cascade-start}
\frac{1-x}{x} \frac{1 - F_1(x)}{1 - F_0(x)} \geq \frac{Q}{1-Q}
\end{align}
\noindent which is also the necessary and sufficient condition that, under $\bm{q}$, up-cascades can start in the information-cascades setting.
\item \label{half} If any player delegates (has threshold $\leq 1 - Q$ when $L < \frac{Q}{1-Q}$), then for all players $j$, $x_j \leq \frac{1}{2}$.
\item \label{time-reverse} In a quality distribution $\bm{q}$, only the earliest players delegate iff for every $x$ satisfying \eqref{cascade-start}
$$(1-x)[1-F_1(x)]^2 - x[1 - F_0(x)]^2 > 1 - 2x$$
\item \label{time-forward} In a quality distribution $\bm{q}$, only the latest players delegate iff for every $x$ satisfying \eqref{cascade-start}
$$(1-x)[1-F_1(x)]^2 - x[1 - F_0(x)]^2 < 1 - 2x$$
\end{enumerate}
\end{theorem}

\begin{proof}
See Appendix.
\end{proof}

For uniformly-distributed qualities (see Section \ref{uniform}), Theorem \ref{delegation}, in conjunction with Proposition \ref{min-R}, shows that players never delegate for $R = \frac{1}{2}$ (exemplified in Figure \ref{fig:5080}), but may delegate starting at some minimum $R$ (as shown in Figures \ref{fig:6580}-\ref{fig:6580-4}). This means that, at least for uniform distributions, players delegate only when they can expect future players to have sufficiently accurate signals.

Quality distributions where delegation is by early players only (following Theorem \ref{delegation}(\ref{time-reverse})) are common (shown in 3 out 4 of Figures \ref{fig:5080}-\ref{fig:6580-4}), and exhibit a {\em time reversal} of sorts: While in the information cascades setting an up-cascade might start after several players invested, in the sequential fundraising game players pre-empt and avert this possibility by delegating. Whenever a player's approval with given threshold $x$ would trigger an up-cascade in the information cascades setting, in our fundraising setting all {\em previous} players would delegate before such a strategy $x$ player. Furthermore, they have ``reverse'' cascades, wherein herding takes place from some player {\em backwards}.

\begin{corollary}[Reverse cascade]
\label{reverse}
Where Theorem \ref{delegation}(\ref{time-reverse}) holds, if $\sigma(L, B,  n) \leq 1 - Q$ then $\sigma(L, B + k,  n + k) \leq 1 - Q$ for every positive integer $k$.
\end{corollary}

Another consequence is that up-cascades cannot start before the fundraising completes.

\begin{corollary}[No Cascades]
\label{no-cascades}
Where Theorem \ref{delegation}(\ref{time-reverse}) holds, in equilibrium, an up-cascade does not start before all players played.
\end{corollary}

\subsubsection{When are Thresholds Equal?}

Consider the unanimity game $B=n$ played {\em simultaneously} rather than sequentially. All players are in a perfectly-symmetrical situation: They need to take a decision, based solely on the prior public likelihood $L_0$. The project will take place iff all players simultaneously approve.

Therefore there exists a symmetric MPE in the simultaneous game, in which all play strategies with the same threshold $x_1 = \ldots = x_n$. This equilibrium may not be unique.

The sequential game is apparently different, in that players observe the decisions of others and adjust their action accordingly. In fact, in a unanimity game, this difference is an illusion: Every player may assume that others invested, since, if they did not, their own decision does not matter. Observing that others, in fact, invested, changes nothing.

This shows why every MPE of the sequential game is an MPE of the simultaneous game (but not vice versa, as Proposition \ref{maximal} is valid only for the sequential game).

This explains why players may play different thresholds, but it also explains why often all or some play the same strategy: When the symmetric equilibrium is the only MPE of the simultaneous game, it must also be the only MPE of the sequential game.

In fact, we can formulate a necessary condition for players to play different strategies.

\begin{proposition} [Equal Thresholds]
\label{equal}
Define $$J(x) := \frac{x}{1-x} \frac{1 - F_0(x)}{1 - F_1(x)}$$\footnote{$J(x)$ is often referred to as the {\em hazard rate} in survival analysis and in some relevant literature, e.g., in \cite{herrera2012necesssary}.}  In MPE, two players play unequal strategies $x_i \neq x_j$ only if $J(x)$ is not strictly monotonic between $x_i$ and $x_j$.
Equivalently, this can happen only when $$\frac{d \log J(x)}{dx} = \frac{1}{x} + \frac{1}{1-x} + \frac{f_1(x)}{1 - F_1(x)} - \frac{f_0(x)}{1 - F_0(x)}$$ changes sign between $x_i$ and $x_j$.
\end{proposition}

\begin{proof}
See Appendix.
\end{proof}

%
%

\subsection{Bounds}

\subsubsection{Threshold Bounds}

Can we bound the equilibrium threshold given the state ($L, B, n$)? Inferences from a player's action increase with the player's threshold (Lemma \ref{monotone}), so if enough players play a high-enough threshold, an up-cascade ensues. This, apparently, gives us a method to bound the equilibrium thresholds. However, the phenomenon of delegation, from which no inference can be made, foils this. The phenomenon of ``reverse cascades'' (Corollary \ref{reverse}) shows that this holds with any number of players: No matter how large $n$ is, all but a few players may delegate, forcing the remaining few players to play strategy thresholds that do not depend on $n$. 

Some improvement to this is in Theorem \ref{delegation}(\ref{half}), where we showed that when any player delegates in a unanimity game, {\em all} players play a threshold less than $1/2$, meaning that they {\em always} invest with a $1$-signal, and invest with a $0$-signal if their signal quality is {\em below} a threshold.

In games where none delegate, players indeed play increasingly lower thresholds with increasing $n$. We exemplify this with unanimity games ($B=n$) where all play the same threshold (which, as shown in Section \ref{uniform}, is a common occurrence). Mark the common threshold by $x$. Then \eqref{threshold} becomes
\begin{align*}
L\frac{x}{1-x} = \Big(\frac{1-F_0(x)}{1-F_1(x)}\Big)^{n-1}
\end{align*}

Since $L\frac{x}{1-x} > L\frac{1-Q}{Q}$, we get
\begin{align}
\label{bound}
\frac{1-F_1(x)}{1-F_0(x)} < \Big(\frac{1}{L}\frac{Q}{1-Q}\Big)^{\frac{1}{n-1}}
\end{align}

The left-hand side of \eqref{bound} increases monotonically (Lemma \ref{monotone}) from $1$ to $\frac{Q}{1-Q}$ when $x$ ranges from $1-Q$ to $Q$, while the right-hand side decreases with $n$ (as $\frac{1}{L}\frac{Q}{1-Q} > 1$), limiting at $1$. Therefore \eqref{bound} gives us an increasingly lower upper-bound on the threshold with increasing $n$, and shows that thresholds converge at $1 - Q$ as the number of players grows. In fact, the following proposition describes the asymptotic behavior of $x$ and of the probabilities of completion for large $n$.

\begin{proposition}[Large $n$ Behavior]
\label{large-n}
Let $B=n$ and suppose that, in equilibrium, all players have the same strategy $x \in (1-Q, Q)$. Then
\begin{enumerate}
\item \label{lg1} If $f_{\bm{q}}(Q) > 0$ $$x = 1 - Q +  \frac{\Big(\frac{1}{L}\frac{Q}{1-Q}\Big)^{\frac{1}{n-1}} - 1}{(2Q - 1) f_{\bm{q}}(Q)} (1 \pm o(1))$$
\item \label{lg2} If $f_{\bm{q}}(Q) = 0$ and $f'_{\bm{q}}(Q) \neq 0$, where $f'_{\bm{q}}(Q)$ denotes the left derivative of $f_{\bm{q}}(x)$ at $Q$  $$x = 1 - Q +  \sqrt{\frac{\Big(\frac{1}{L}\frac{Q}{1-Q}\Big)^{\frac{1}{n-1}} - 1}{(2Q - 1) \frac{-f'_{\bm{q}}(Q)}{2}}} (1 \pm o(1))$$
\item \label{lg3} $$\pi_0(L, n, n) = \Big(\frac{1}{L}\frac{Q}{1-Q}\Big)^{-\frac{Q}{2Q-1}}  (1 \pm o(1))$$
\item \label{lg4} $$\pi_1(L, n, n) = \Big(\frac{1}{L}\frac{Q}{1-Q}\Big)^{-\frac{1-Q}{2Q-1}}  (1 \pm o(1))$$
\end{enumerate}
\end{proposition}

\begin{proof}
See Appendix.
\end{proof}

We remark about Proposition \ref{large-n} that it shows that the large-$n$ behavior depends only on $Q$ and the quality density in its vicinity, but we remind that it holds only subject to the no-delegation assumption, which, as we have seen, depends on the entire quality distribution, and its lower-bound $R$ in particular.

\subsubsection{Total Probability of Completion}

We define the {\em total} probability of completion $\pi(L, B, n)$ as the probability that the funding will complete given the state $(L, B, n)$. Since the probability that $\omega = 1$ is $\frac{L}{1+L}$, we have
\begin{align}
\pi(L, B, n) = \frac{L}{1 + L} \pi_1(L, B, n) + \frac{1}{1+L} \pi_0(L, B, n)
\end{align}

From Corollary \ref{recurrence}, $\pi(L, B, n)$ is well-defined and continuous in $L$ except at irregular states.

Using Proposition \ref{large-n}, we can calculate the large-$n$ behavior of $\pi(L, n, n)$ with no delegation. Its value, for $Q = 0.8$, is plotted in Figure \ref{fig:completion}. The probability rises monotonically from $0$ to $1$ as $L$ ranges from $0$ to the up-cascade limit. For $Q = 0.8$, the probability for completion is below $50\%$ when the probability for $\omega=1$ is below $\simeq 57\%$.

\subsubsection{Probability of Correctness}

We define the probability of correctness $\gamma(L, B, n)$ as the probability that the outcome of the funding will match $\omega$ given the state $(L, B, n)$, i.e., complete when $\omega = 1$, and fail to complete when $\omega = 0$. Since the probability that $\omega = 1$ is $\frac{L}{1+L}$, we have
\begin{align}
\gamma(L, B, n) = \frac{L}{1 + L} \pi_1(L, B, n) + \frac{1}{1+L} \big[1 - \pi_0(L, B, n)\big]
\end{align}

From Corollary \ref{recurrence}, $\gamma(L, B, n)$ is well-defined and continuous in $L$ except at irregular states.

Using Proposition \ref{large-n}, we can calculate the large-$n$ behavior of $\gamma(L, n, n)$ with no delegation. Its value, for $Q = 0.8$, is plotted in Figure \ref{fig:correctness}. Note that, for low likelihoods, the correctness is close to $1$, since the fundraising is almost certain to fail, as it most probably should. At  the up-cascade threshold of $L=\frac{Q}{1-Q}$, the probability of correctness is $Q$, since the fundraising is certain to complete, while the probability that this is the right outcome is $Q$. The correctness is lowest ($\simeq 72\%$) when the probability for $\omega = 1$ is $\simeq 60\%$ (for $Q=0.8$).

\begin{figure}[tb]
\centering
\begin{minipage}{.5\textwidth}
  \centering
  \includegraphics[width=1\linewidth]{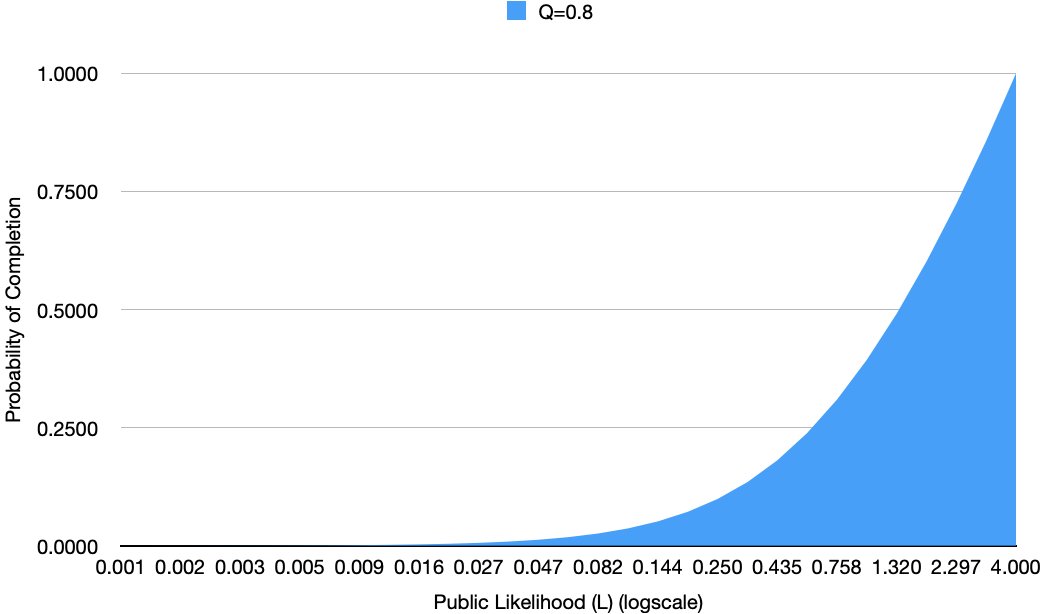}
  \caption{Probability of completion $\pi(L,\infty,\infty)$}
  \label{fig:completion}
\end{minipage}%
\begin{minipage}{.5\textwidth}
  \centering
  \includegraphics[width=1\linewidth]{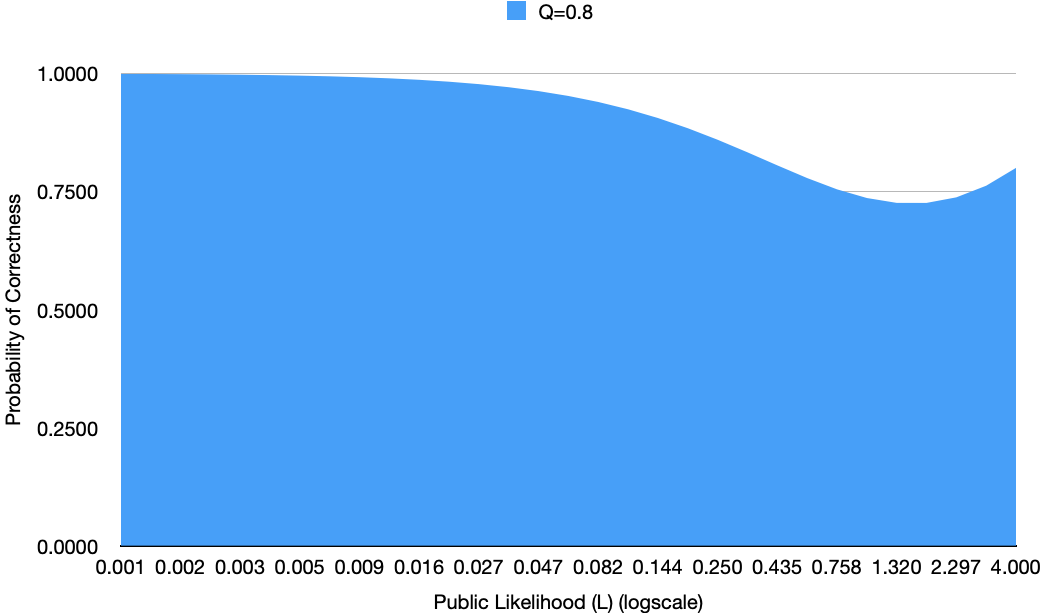}
  \caption{Probability of correctness $\gamma(L,\infty,\infty)$}
  \label{fig:correctness}
\end{minipage}
\end{figure}

\subsection{The General Case: Non-Unanimous Decisions}

\begin{table}
\s
\center
\caption{Optimal Thresholds for $n = 5, B = 3$ with $\bm{q} \sim U(0.65,0.8)$}
\label{tab:n5B3}
 \begin{tabular}{|c|r|r|r|r|r|} 
 \hline
        \textbf{Likelihood} & $x_1$ & $x_2$ & $x_3$ & $x_4$ & $x_5$\\
        \hline
0.5000& 0.445144 [1]& 0.412324 [1]& 0.223453 [1] &&\\
      &              &              &           \dr\:[0]& 0.517217 [1] & \\
      &              &              &              &           \dr\:[0]& 0.738521 [1] \\
      &              &           \dr\:[0]& 0.497065 [1]& 0.431373 [1] &\\
      &              &              &              &           \dr\:[0]& 0.666667 [1] \\
      &              &              &           \dr\:[0]& 0.660662 [1]& 0.660662 [1] \\
      &           \dr\:[0]& 0.532386 [1]& 0.497065 [1]& 0.431373 [1] &\\
      &              &              &              &           \dr\:[0]& 0.666667 [1] \\
      &              &              &           \dr\:[0]& 0.660662 [1]& 0.660662 [1] \\
      &              &           \dr\:[0]& 0.657836 [1]& 0.657836 [1]& 0.657836 [1] \\
     \hline
      1.0000& 0.301337 [1]& 0.372818 [1]& Cascade [1] & &\\
&      &                       \dr\:[0]& 0.445748 [1]& 0.372818 [1] & \\
&      &              &              &                         \dr\:[0]& 0.610461 [1] \\
&      &              &                        \dr\:[0]& 0.610461 [1]& 0.610461 [1] \\
&                \dr\:[0]& 0.463950 [1]& 0.430168 [1]& 0.300745 [1] \\
&      &              &              &                         \dr\:[0]& 0.562884 [1] \\
&      &              &                        \dr\:[0]& 0.531371 [1]& 0.531371 [1] \\
&          &                     \dr\:[0]& 0.531371 [1]& 0.531371 [1]& 0.531371 [1]  \\
\hline
2.0000 & Delegate [1] & 0.333333 [1] & Cascade [1]& & \\
          &  & \dr\:[0] & 0.423937 [1] & 0.342776 [1] & \\
            &  &           &            &            \dr\:[0] & 0.583375 [1] \\
                &        & &                     \dr\:[0] & 0.578947 [1] & 0.578947 [1] \\
                \hline
\end{tabular}
\subcaption*{{\em $a=1$ actions continue to the right. $a=0$ actions continue down to the \dr arrow and to the right.}}
\subcaption*{$x_i$ = $i$'th contributor's threshold in equilibrium.}
\n
\end{table}

\begin{table}
\s
\center
\caption{Optimal Thresholds for $n = 6, B = 4$ with $\bm{q} \sim U(0.65,0.8)$}
\label{tab:n6B4}
 \begin{tabular}{|c|r|r|r|r|r|r|} 
 \hline
        \textbf{Likelihood} & $x_1$ & $x_2$ & $x_3$ & $x_4$ & $x_5$ & $x_6$\\
        \hline
2.0000& Delegate [1]& Delegate [1]& 0.333333 [1]& Cascade [1] && \\
&      &              &                         \dr\:[0]& 0.423937 [1]& 0.342776 [1] &\\
&      &              &              &              &                         \dr\:[0]& 0.583375 [1]\\
&      &              &              &                         \dr\:[0]& 0.578947 [1]& 0.578947 [1]\\
\hline
\end{tabular}
\subcaption*{{\em $a=1$ actions continue to the right. $a=0$ actions continue down to the \dr arrow and to the right.}}
\subcaption*{$x_i$ = $i$'th contributor's threshold in equilibrium.}
\n
\end{table}

In the general case, $B < n$, the fundraising will succeed if at most $n-B$ contributors decline to invest. There are therefore $n \sum_{i = 0}^{n-B} {n \choose i}$ subgames with a positive probability for completion.

Based on the equilibrium characterization (Corollary \ref{c}), a solution for equilibrium requires the solution of simultaneous equations, one for each subgame, each formed as \eqref{threshold}. In Section \ref{B=1} we provided a complete solution of the case $B \leq 1$, and in Section \ref{B=n} we provided a thorough analysis of the case $B = n$. A similar treatment of $1 < B < n$ is at present beyond our reach, though we have provided a characterization of its cascades (Section \ref{cascades}), and a recurrence relation for the probability of completion (Section \ref{recur}). 

A numerical solution for equilibrium strategies must therefore rely on dynamic programming. 

A stumbling block in the analysis is our inability to prove the following, which we conjecture.

\begin{conjecture} [$\pi_1 \geq \pi_0$]
\label{pi-conj}
Except at irregular states $(L, B, n)$, for every public likelihood $L$, $B$ investments outstanding and $n$ players remaining $$\pi_1(L, B, n) \geq \pi_0(L, B, n)$$ with equality only in a cascade or for $B \leq 0$. 
\end{conjecture}


Conjecture \ref{pi-conj} has an immediate consequence on mutual insurance, namely that in {\em every} state in a sequential fundraising players assume mutual insurance. Indeed, wherever the conjecture is true, \eqref{threshold} translates into $L \frac{x}{1-x} \leq 1$, with equality only for $B \leq 1$ or in a cascade. On the other hand, with the same public likelihood, the last player's equilibrium strategy $x'$ satisifies $L \frac{x'}{1-x'} = 1$. Therefore $x \leq x'$.

\begin{corollary} [Mutual Insurance]
\label{insurance}
Let $x_i$ be the strategy of player $i \in [n]$ in MPE with prior likelihood $L$, $B$ investments outstanding and $n$ players remaining. If Conjecture \ref{pi-conj} holds, i.e., if $\pi_1(L, B, n) \geq \pi_0(L, B, n)$, then $x_i \leq \sigma(L, 1, 1)$ with equality only for the last player or in an up-cascade.
\end{corollary}

\subsubsection{Example Games}

Tables \ref{tab:n5B3} and \ref{tab:n6B4} illustrate all equilibrium strategies for two sample games, the first requiring 3 investments from 5 contributors, the second requiring 4 investments from 6 contributors. In both, qualities are uniformly distributed from $R=0.65$ to $Q=0.8$.

Optimal thresholds are given in Table \ref{tab:n5B3} for $L = 0.5, 1$ and $2$, i.e., for $\Pr[\omega=1] = \frac{1}{3}, \frac{1}{2}$ and $\frac{2}{3}$. In the latter case ($L=2$), the first player delegates, herding on investment. Then, if the second player's type is at least $1/3$, he invests, and an up-cascade leads to the third required investment. Otherwise, the game continues as shown in the table.

In Table \ref{tab:n6B4} optimal thresholds are given for $L=2$ only, and illustrate the reverse cascade effect: The first player's delegation ``continues'' the first player's delegation in Table \ref{tab:n5B3} for the same likelihood.

\section{Sequential Voting}
\label{voting}

Here we consider a slightly modified model, that captures sequential voting better than sequential fundraising. A resolution is passed if $B$ or more voters vote ``aye'', and voters are motivated to pass ``good'' resolutions and vote down ``bad'' resolutions. If the resolution does not pass, no one gains or loses. While  in a sequential fundraising those who do not contribute have zero utility, regardless of the outcome, in a sequential vote, voters who say ``nay'' (equivalent to declining in a fundraising) have the same payoff structure as those who say ``aye'' (equivalent to contributing in a fundraising).

As already noted, in unanimity games, there is no difference between sequential fundraising and sequential voting. Since, as we have shown, mutual insurance and delegation are features of unanimity games, they are also features of general voting games.

Based on the model in Section \ref{model}, the {\em only} modification we make is that {\em all} players get utility $2 \omega - 1$ (rather than only players $i$ who played $a_i = 1$).

Retracing our analysis, Section \ref{recur} is unchanged. The probability of "completion" now signifies the probability of passing the resolution. Inferences from actions (Section \ref{inferences}) remain unchanged.

Equilibrium strategies are, again, threshold strategies. To see this, refer to the proof of Theorem \ref{threshold-strategy}: Equation \ref{expectation-phi} remains the player's expectation if she votes ``aye'', but her expectation for voting ``nay'' is no longer $0$ but instead
\begin{align}
V^{\phi_1}(L, B, n, t) &= \Pr[\omega = 1 | L, t] \pi_1(L^-, B, n - 1) - \Pr[\omega = 0 | L, t] \pi_0(L^-, B, n - 1) \nonumber \\
\label{expectation-phiV}
&= \Big(1 - \frac{1}{1 + L \frac{t}{1-t}}\Big) \pi_1(L^-, B, n - 1) - \frac{1}{1 + L \frac{t}{1-t}} \pi_0(L^-, B, n- 1)
\end{align}
So the player prefers voting ``aye'' to ``nay'' iff
\begin{align}
\label{invest-V}
L \frac{t}{1-t} \geq \frac{\pi_0(L^+, B - 1, n - 1) - \pi_0(L^-, B,  n-1)}{\pi_1(L^+, B - 1, n- 1) - \pi_1(L^-, B, n-1)}
\end{align}

Since the right-hand side of \eqref{invest-V} does not depend on $t$, the rest of the proof of Theorem \ref{threshold-strategy} applies, with the conclusion (the counterpart of Corollary \ref{indifference}) that equilibrium strategies are threshold strategies, the threshold $x$ satisfying
\begin{align}
\label{threshold-V}
L \frac{x}{1-x} = \frac{\pi_0\big(L \frac{1 - F_1(x)}{1 - F_0(x)}, B - 1, n - 1\big) - \pi_0\big(L \frac{F_1(x)}{F_0(x)}, B, n - 1\big)}{\pi_1\big(L \frac{1 - F_1(x)}{1 - F_0(x)}, B - 1, n- 1\big) - \pi_1\big(L \frac{F_1(x)}{F_0(x)}, B, n- 1\big)}
\end{align}

Note that, in unanimity voting games, $B=n$, the terms $ \pi_\omega\big(L \frac{F_1(x)}{F_0(x)}, B, n - 1\big)$ are equal to $0$, so that \eqref{threshold-V} is identical to \eqref{threshold}, the sequential fundraising equilibrium condition.

Similarly there are modified versions of \eqref{pi_diff}, the {\em type-independent utility expectation}, and \eqref{discriminator}, the {\em discriminator}, for tie-breaking multiple solutions of \eqref{threshold-V}, and a modified definition of irregular states, analogous to Corollary \ref{irregular}.

However, in sequential voting, the case $B = 1$, when a single ``aye'' is needed to pass the resolution, is {\em not} equivalent to an information cascades setting. (Observe that, the right-hand side of \eqref{threshold-V}, unlike that of \eqref{threshold}, is not equal to $1$ when substituting $B = 1$). Solving the case $B=1$ for equilibrium strategies is already non-trivial. The case $B \leq 0$ (i.e., when the resolution has already passed) is equivalent to information cascades.

The herd-invest condition is, from \eqref{threshold}
\begin{align}
\label{herd-invest-V}
L \frac{1-Q}{Q} \geq \frac{\pi_0\big(L, B - 1, n - 1\big) - \pi_0\big(L \frac{1-Q}{Q}, B, n - 1\big)}{\pi_1\big(L, B - 1, n- 1\big) - \pi_1\big(L \frac{1-Q}{Q}, B, n- 1\big)}
\end{align}

\noindent and the herd-decline condition is
\begin{align}
\label{herd-decline-V}
L \frac{Q}{1-Q} \leq \frac{\pi_0\big(L \frac{Q}{1-Q}, B - 1, n - 1\big) - \pi_0\big(L, B, n - 1\big)}{\pi_1\big(L \frac{Q}{1-Q}, B - 1, n- 1\big) - \pi_1\big(L, B, n- 1\big)}
\end{align}

In a down-cascade, all remaining players vote ``nay''. Applying \eqref{herd-decline-V} repeatedly, we get, as for sequential fundraising, a down-cascade for $L \leq \Big(\frac{1-Q}{Q}\Big)^B$.

In an up-cascade, all remaining players vote ``aye''. Applying \eqref{herd-invest-V} repeatedly, clearly an up-cascade is in progress for $L \geq \Big(\frac{Q}{1-Q}\Big)^{n-B+1}$, but this is not necessarily tight.

\section{Discussion}
\label{discussion}

\subsection{Conclusions}

We presented and analyzed a model of sequential fundraising, and formulate equations satisfied in equilibrium. We demonstrated, among else, the phenomena of mutual insurance and delegation, that threshold strategies are played, and the frequent equality of thresholds. We showed that this applies to sequential voting, too. Interesting game-theoretical aspects include a relation of the solution to the uniqueness of the simultaneous-game equilibrium.

The phenomena show that the results of a sequential fundraising or vote should be taken with a grain of salt, especially for successful ones (completion in a fundraising, and adoption in a vote). Even a massive up vote is not proof that the resolution in question is solid. This holds true even when some late participants decline, proving that no cascade is in progress. Special skepticism should be directed towards the proposition that actions of early participants reflect their judgement, as they may be shading it higher, or completely abdicating their role in voicing their true opinion. 

\subsection{The Importance of Non-Uniform Qualities}
\label{non-uniform}

The study of information cascades started with the assumption of equally-informed players, then progressed to signal models which translate, in our model, to non-uniform qualities from a continuous range. In that context, this sought to generalize a simplifying assumption to a more credible one. In our sequential fundraising model, distancing from uniform qualities is of more fundamental importance, as equilibrium behavior becomes progressively more chaotic, and of dubious economical significance, as the quality range is more restricted. Compare, e.g., the behavior in Figure \ref{fig:7580} to that in Figures \ref{fig:5080} and \ref{fig:6580}. 

\subsection{Future Work}

In the general fundraising setting ($B \leq n$) many questions remain open. Our Conjecture \ref{pi-conj} and its Corollary \ref{insurance} await settlement. Some properties of delegation (Theorem \ref{delegation}) are most probably true in the general case, perhaps all of them.

We did not investigate some questions, most of which seem difficult: Do players always play different thresholds when such an equilibrium is playable? Players often, but not always, play increasingly-high thresholds according to their order (See Figures \ref{fig:5080}-\ref{fig:6580-4}). What rule guides this? 
An efficient dynamic programming solution will be useful in investigating the general problem.

A different set of questions is raised when contributors are not statistically equivalent, whether they come from different quality distributions, or whether their budgets for contribution are different. The main question is: Given a choice, in what order should the fundraiser approach the potential contributors? This, of course, introduces a new kind of player, and raises the complexity of the analysis to another level.

In our model, the timing of decisions is up to the fundraiser. What happens if the timing is endogenous? The contributors observe all actions (and failures to act), and invest, if they so decide, at a time of their choosing. Assume time is punctuated by rounds. Clearly, after a finite number of rounds have passed without action, the fundraising can be terminated. This leads to very different reasoning on the contributors' part. This variation more closely models crowdfunding than seed fundraising.

\bibliographystyle{apalike}
\bibliography{crowdfunding}

\begin{thebibliography}{}

\bibitem[Acemoglu et~al., 2011]{Acemoglu2011}
Acemoglu, D., Daleh, M.~A., Lobel, I., and Ozdaglar, A. (2011).
\newblock {Bayesian Learning in Social Networks}.
\newblock {\em Review of Economic Studies}, 78:1201--1236.

\bibitem[Alaei et~al., 2016]{Alaei2016}
Alaei, S., Malekian, A., and Mostagir, M. (2016).
\newblock {Working Paper A Dynamic Model of Crowdfunding}.
\newblock In {\em Proceedings of the 2016 ACM Conference on Economics and
  Computation}, Maastricht, The Netherlands. ACM.

\bibitem[Ali and Kartik, 2012]{Ali2012}
Ali, S.~N. and Kartik, N. (2012).
\newblock {Herding with collective preferences}.
\newblock {\em Economic Theory}, 51(3):601--626.

\bibitem[Arieli et~al., 2018]{arieli2018one}
Arieli, I., Koren, M., and Smorodinsky, R. (2018).
\newblock {The One-Shot Crowdfunding Game}.
\newblock In {\em EC '18: Proceedings of the 2018 ACM Conference on Economics
  and Computation}, pages 213--214. ACM.

\bibitem[Arieli et~al., 2019a]{Arieli2018}
Arieli, I., Koren, M., and Smorodinsky, R. (2019a).
\newblock {Information aggregation in large collective decisions.}
\newblock {\em Working Paper}.

\bibitem[Arieli et~al., 2019b]{Arieli2018a}
Arieli, I., Koren, M., and Smorodinsky, R. (2019b).
\newblock {The Implications of Pricing on Social Learning}.
\newblock {\em EC '19: Proceedings of the 2019 ACM Conference on Economics and
  Computation}.

\bibitem[Ban and Mansour, 2018]{ban2018all}
Ban, A. and Mansour, Y. (2018).
\newblock Are all experts equally good? a study of analyst earnings estimates.
\newblock {\em arXiv preprint arXiv:1806.06654}.

\bibitem[Banerjee, 1992]{Banerjee1992}
Banerjee, A. (1992).
\newblock {A Simple Model of Herd Behavior}.
\newblock {\em The Quarterly Journal of Economics}, 107(3):797--817.

\bibitem[Bergemann and Hege, 2004]{Bergemann2004}
Bergemann, D. and Hege, U. (2004).
\newblock {The financing of innovation: Learning and stopping}.
\newblock {\em Cowles Foundation Discussion}, 1292R(4):14--17.

\bibitem[Bergemann et~al., 2008]{Bergemann2008}
Bergemann, D., Hege, U., and Peng, L. (2008).
\newblock {Cowles Foundation}.
\newblock {\em Cowles Foundation for Research in Economics at Yale University,
  Discussion Paper}, 1682(1398).

\bibitem[Bikhchandani et~al., 1992]{bikhchandani1992theory}
Bikhchandani, S., Hirshleifer, D., and Welch, I. (1992).
\newblock {A theory of fads, fashion, custom, and cultural change as
  informational cascades}.
\newblock {\em Journal of political Economy}, 100(5):992--1026.

\bibitem[Callander, 2007]{Callander2007}
Callander, S. (2007).
\newblock {Bandwagons and Momentum in Sequential Voting}.
\newblock {\em The Review of Economic Studies}, 74(3):653--684.

\bibitem[Dekel and Piccione, 2000]{Dekel2000}
Dekel, E. and Piccione, M. (2000).
\newblock {Sequential voting procedures in symmetric binary elections}.
\newblock {\em Journal of Political Economy}, 108(1):34--55.

\bibitem[Ely, 2017]{Ely2017}
Ely, J.~C. (2017).
\newblock {Beeps}.
\newblock {\em The American Economic Review}, 107(1):31--53.

\bibitem[Eyster et~al., 2013]{Eyster2013}
Eyster, E., Galeotti, A., Kartik, N., and Rabin, M. (2013).
\newblock {Congested Observational Learning}.
\newblock {\em Games and Economic Behavior}, 87(283454):1--32.

\bibitem[Goel et~al., 2010]{goel2010prediction}
Goel, S., Reeves, D.~M., Watts, D.~J., and Pennock, D.~M. (2010).
\newblock Prediction without markets.
\newblock In {\em Proceedings of the 11th ACM conference on Electronic
  commerce}, pages 357--366. ACM.

\bibitem[Halac et~al., 2020]{Halac2020}
Halac, M., Kremer, I., and Winter, E. (2020).
\newblock {Raising capital from heterogeneous investors†}.
\newblock {\em American Economic Review}, 110(3):889--921.

\bibitem[Hellmann and Thiele, 2015]{Hellmann2015}
Hellmann, T. and Thiele, V. (2015).
\newblock {Friends or foes? The interrelationship between angel and venture
  capital markets}.
\newblock {\em Journal of Financial Economics}, 115(3):639--653.

\bibitem[Herrera and H{\"{o}}rner, 2012]{herrera2012necesssary}
Herrera, H. and H{\"{o}}rner, J. (2012).
\newblock {A Necesssary and Sufficient Condition for Information Cascades}.

\bibitem[Herrera and H{\"{o}}rner, 2013]{Herrera2013}
Herrera, H. and H{\"{o}}rner, J. (2013).
\newblock {Biased Social Learning}.
\newblock {\em Games and Economic Behavior}, 80:131--146.

\bibitem[Maskin and Tirole, 2001]{maskin2001markov}
Maskin, E. and Tirole, J. (2001).
\newblock {Markov perfect equilibrium: I. Observable actions}.
\newblock {\em Journal of Economic Theory}, 100(2):191--219.

\bibitem[Mollick, 2014]{Mollick2014}
Mollick, E.~R. (2014).
\newblock {The dynamics of crowdfunding: An exploratory study}.
\newblock {\em Journal of Business Venturing}, 29(1):1--16.

\bibitem[Moscarini and Ottaviani, 2001]{Moscarini2001a}
Moscarini, G. and Ottaviani, M. (2001).
\newblock {Price competition for an informed buyer}.
\newblock {\em Journal of Economic Theory}, 101(2):457--493.

\bibitem[Mueller-frank, 2012]{Mueller-frank2012}
Mueller-frank, M. (2012).
\newblock {Market Power, Fully Revealing Prices and Welfare}.

\bibitem[Smith and S{\o}rensen, 2000]{Smith2012}
Smith, B. Y.~L. and S{\o}rensen, P. (2000).
\newblock {Pathological Outcomes of Observational Learning.}
\newblock {\em Econometrica}, 68(2):371--398.

\bibitem[Smith et~al., 2017]{Smith2017exp}
Smith, L., S{\o}rensen, P.~N., and Tian, J. (2017).
\newblock {Informational Herding, Optimal Experimentation, and Contrarianism *
  †}.
\newblock Technical report.

\bibitem[Strausz, 2017]{Strausz2017}
Strausz, R. (2017).
\newblock {A Theory of Crowdfunding - a mechanism design approach with demand
  uncertainty and moral hazard}.
\newblock {\em American Economic Review}, 107(6):1--40.

\end{thebibliography}

\newpage
\appendix

\section{Proof of Proposition \ref{equally-optimal}}

\begin{proof}
Suppose the player plays $\mathcal{S}(x_i)$ with probability $1 - \epsilon$, and $\mathcal{S}(x_j)$ with probability $\epsilon$. This mixed strategy induces a posterior likelihood $L^+(\epsilon)$ which limits at the unmixed posterior $L \frac{1 - F_1(x_i)}{1 - F_0(x_i)}$ as $\epsilon \to 0$. Similarly, $\pi_\omega(L^+(\epsilon), B-1, n-1)$ limits at $\pi_\omega\big(L \frac{1 - F_1(x_i)}{1 - F_0(x_i)}, B - 1, n - 1\big)$ as $\epsilon \to 0$, since $\pi_\omega$ is continuous except at irregular states (Corollary \ref{recurrence}).

The pure strategies have the same utility $U$, so,
\s\begin{align*}
U =&  \frac{L}{1 + L}[1 - F_1(x_i)] \pi_1\big(L \frac{1 - F_1(x_i)}{1 - F_0(x_i)}, B - 1, n - 1\big) -  \frac{1}{1 + L}[1 - F_0(x_i)] \pi_0\big(L \frac{1 - F_0(x_i)}{1 - F_0(x_i)}, B - 1, n - 1\big) = \\
=&  \frac{L}{1 + L}[1 - F_1(x_j)] \pi_1\big(L \frac{1 - F_1(x_j)}{1 - F_0(x_j)}, B - 1, n - 1\big) -  \frac{1}{1 + L}[1 - F_0(x_j)] \pi_0\big(L \frac{1 - F_0(x_j)}{1 - F_0(x_j)}, B - 1, n - 1\big)
\end{align*}\n

The utility of the mixed strategy, $U_{i,j}(\epsilon)$, is
\s\begin{align*}
U_{i,j}(\epsilon) =  \frac{L}{1 + L}\big\{(1-\epsilon)[1 - F_1(x_i)] &+ \epsilon [1 - F_1(x_j)]\big\}\pi_1\big(L \frac{1 - F_1(x_i)}{1 - F_0(x_i)}, B - 1, n - 1\big) \\
&-  \frac{1}{1 + L}\big\{(1-\epsilon)[1 - F_0(x_i)] + \epsilon [1 - F_0(x_j)]\big\}\ \pi_0\big(L \frac{1 - F_0(x_i)}{1 - F_0(x_i)}, B - 1, n - 1\big) \\
 = (1 - \epsilon)& U + \frac{\epsilon}{1+L} \big\{L[1 - F_1(x_j)] \pi_1(L^+(\epsilon), B - 1, n - 1) - [1 - F_0(x_j)] \pi_0(L^+(\epsilon), B - 1, n - 1)\big\}
\end{align*}\n

As $\epsilon \to 0$, $\frac{\max_{j \neq i} U_{i,j}(\epsilon) - U}{\epsilon} \to \frac{U + D(x_i)}{1 + L}$. Therefore, if $D(x_i) > D(x_k)$, for a sufficiently small $\epsilon$ the player prefers to maximize the probability of $\mathcal{S}(x_i)$, choosing to mix it with the strategy of player $\arg\max_{j \neq i} U_{i,j}(\epsilon)$, rather than maximize the probability of $\mathcal{S}(x_k)$, mixed with any strategy.
\end{proof}

%
%

\section{Proof of Lemma \ref{monotone}}
\begin{proof}
We shall prove the lemma by showing that the differential of the log of each expression is positive.

Define $\phi(x) := f_0(x) / (1-x) = f_1(x) / x$. We show that, for $x \in (1 - Q, Q)$
\begin{align}
\label{ineq-f}
\frac{1 -F_0(x)}{1-x} < \frac{1 - F_1(x)}{x}
\end{align}

For
\begin{align*}
\frac{1 -F_0(x)}{1-x} - \frac{1 - F_1(x)}{x} &= \frac{1}{1-x}\int\limits_x^Q f_0(y) dy - \frac{1}{x}\int\limits_x^Q f_1(y) dy \\
&= \int\limits_x^Q \frac{1 - y}{1 - x} \phi(y) dy - \int\limits_x^Q\frac{y}{x} \phi(y) dy \\
&=  \int\limits_x^Q \Big(\frac{1 - y}{1 - x} - \frac{y}{x}\Big) \phi(y) dy = \int\limits_x^Q \frac{x - y}{x(1 - x)} \phi(y) dy < 0
\end{align*}

\noindent because the integrand is negative for $0 < x < y < 1$.

Therefore
\begin{align*}
\frac{d}{dx} \log \frac{1 - F_1(x)}{1 - F_0(x)} = \frac{f_0(x)}{1 - F_0(x)} - \frac{f_1(x)}{1 - F_1(x)} = \phi(x) \Big[\frac{1 - x}{1 - F_0(x)} - \frac{x}{1 - F_1(x)}\Big] > 0
\end{align*}
\noindent with the last inequality following from \eqref{ineq-f}.

Similarly, we show that, for $x \in (1 - Q, Q)$
\begin{align}
\label{ineq-f1}
\frac{F_0(x)}{1-x} > \frac{F_1(x)}{x}
\end{align}

For
\begin{align*}
\frac{F_0(x)}{1-x} - \frac{F_1(x)}{x} &= \frac{1}{1-x}\int\limits_{1-Q}^x f_0(y) dy - \frac{1}{x}\int\limits_{1-Q}^x f_1(y) dy \\
&= \int\limits_{1-Q}^x \frac{1 - y}{1 - x} \phi(y) dy - \int\limits_{1-Q}^x\frac{y}{x} \phi(y) dy \\
&=  \int\limits_{1-Q}^x \Big(\frac{1 - y}{1 - x} - \frac{y}{x}\Big) \phi(y) dy = \int\limits_{1-Q}^x \frac{x - y}{x(1 - x)} \phi(y) dy > 0
\end{align*}

\noindent because the integrand is positive for $0 < y < x < 1$.

Therefore
\begin{align*}
\frac{d}{dx} \log \frac{F_1(x)}{F_0(x)} = \frac{f_1(x)}{F_1(x)} - \frac{f_0(x)}{F_0(x)} = \phi(x) \Big[\frac{x}{F_1(x)} - \frac{1 - x}{F_0(x)}\Big] > 0
\end{align*}
\noindent with the last inequality following from \eqref{ineq-f1}
\end{proof}

\section{Proof of Proposition \ref{min-R}}
\begin{proof}
The inference from investment, substituting \eqref{below-1} and \eqref{below-0} in \eqref{infer-invest} is
\begin{align}
\label{infer-invest1}
L^+(L, x) =  \left\{  \begin{array}{ll}
L \frac{2(Q - R) - x^2 + (1 - Q)^2}{2(Q - R) - Q^2 + (1 - x)^2}  & x \leq 1 - R \\
L \frac{Q + R}{2 - Q - R}& 1-R \leq x \leq R \\
L \frac{Q + x}{2 - Q - x} & x \geq R
\end{array} \right.
\end{align}

By \eqref{last-strategy}, if a cascade is not in progress, a player in the information-cascades setting with public likelihood $L$ has threshold $x$ for which $L = \frac{1-x}{x}$. For such a player's investment to trigger an up-cascade requires $$L \frac{1 - F_1(x)}{1- F_0(x)} \geq \frac{Q}{1-Q}$$ Substituting, an up-cascade needs a solution $x \in (1-Q,Q)$ of
\begin{align}
\label{trigger}
H(x) := \frac{1-x}{x} \frac{1 - F_1(x)}{1- F_0(x)} \geq \frac{Q}{1-Q}
\end{align}

\noindent First, note that, when $R = 1/2$, substituting \eqref{infer-invest1} in \eqref{trigger} yields
\begin{align}
\label{trigger-high}
H(x) = \frac{1-x}{x} \frac{Q+x}{2 -Q - x} \geq \frac{Q}{1-Q}
\end{align}

\noindent which has no solutions in  $(0, 1)$, so there are no up-cascades when $R = 1/2$.

For $x \geq R$, substituting \eqref{infer-invest1} in \eqref{trigger} also results in \eqref{trigger-high}, so cascades can start only for $x < R$. Since $H(1-R) > H(x)$ for every $x \in [1-R, R]$, we may restrict the investigation to $x \leq 1 - R$.

Define $$h(Q) := \frac{1}{1 + \sqrt{\frac{(1-Q)(1+Q)}{Q(2-Q}}}$$ We shall first show that, for every $Q > 1/2$, an up-cascade is possible with $R = h(Q)$. Specifically, a player with threshold $1-R$ will start it. Since, from \eqref{trigger}, $R = h(Q)$ is a solution of
\begin{align}
\label{trigger-low}
H(1-R) = \frac{R}{1-R} \frac{Q+R}{2 -Q - R} = \frac{Q}{1-Q}
\end{align}

Now observe that $\frac{R}{1-R}$ and $\frac{Q+R}{2 -Q - R}$ are both increasing with $R$. Therefore, from \eqref{trigger-low}, for every $R \geq h(Q)$
$$H(1-R) = \frac{R}{1-R} \frac{Q+R}{2 -Q - R} \geq \frac{Q}{1-Q}$$ so that cascades occur (and specifically at $x = 1 - R$) for every $R \geq h(Q)$.

Now we show that cascade do {\em not} occur when $R < h(Q)$.

As noted, $\frac{R}{1-R} \frac{Q+R}{2 -Q - R}$ is increasing with $R$, so that, for $R < h(Q)$, at $x = 1 - R$, $H(1-R) = \frac{R}{1-R} \frac{Q+R}{2 -Q - R} < \frac{Q}{1-Q}$, so a player with $x = 1-R$ does {\em not} start a cascade. Furthermore, $x = 1 - Q$ satisfies \eqref{trigger} with equality $H(1-Q) = \frac{Q}{1 - Q} \frac{1 - F_1(1 - Q)}{1- F_0(1 - Q)} \geq \frac{Q}{1-Q}$. If we could prove that $H(x)$ is convex between $1-Q$ and $1-R$, it would prove that \eqref{trigger} has no solutions in $(1 - Q, 1 - R]$, and therefore no solutions at all for $R < h(Q)$.

By \eqref{infer-invest}, for $x \in [1-Q, 1-R]$ 
\begin{align}
\label{J}
H(x) = \frac{1-x}{x} \frac{2(Q - R) - x^2 + (1 - Q)^2}{2(Q - R) - Q^2 + (1 - x)^2}
\end{align}

For \eqref{J}, $\frac {d^2}{dx^2} H(x) \leq 0$ everywhere in $[1 - Q, 1 - R]$ (according to {\em Wolfram Alpha}). So $H(x)$ is convex, completing the proof for $R < h(Q)$.
\end{proof}

\section{Proof of Theorem \ref{delegation}}
\begin{proof}
An up-cascade can start iff there exists a non-cascade prior likelihood $L \in (\frac{1-Q}{Q}, \frac{Q}{1-Q})$ for which an investment by a player induces the posterior likelihood above the up-cascade threshold $\frac{Q}{1-Q}$. Since when $B \leq 1$ a player's threshold strategy $x$ satisfies $L = \frac{1-x}{x}$ (by \eqref{last-strategy}), and the posterior likelihood after investment is $L\frac{1-F_1(x)}{1-F_0(x)}$, this is conditioned on
\begin{align}
\label{start-cascade}
\frac{1-x}{x} \frac{1-F_1(x)}{1-F_0(x)} \geq \frac{Q}{1-Q}
\end{align}

Since $\frac{1-F_1(x)}{1-F_0(x)} \leq \frac{Q}{1-Q}$ (by Lemma \ref{monotone} and \eqref{limQ}), for \eqref{start-cascade} to hold we must have $\frac{1-x}{x} \geq 1$, or $x \leq \frac{1}{2}$.

Let $x_i$ be a player who delegates, and $x_j$ a player who does not. Then, from \eqref{una-threshold}
\begin{align*}
L \frac{1-Q}{Q} &\geq L \frac{x_i}{1 - x_i} = \prod_{k \in [n], k \neq i} \frac{1- F_0(x_k)}{1- F_1(x_k)} = \prod_{k \in [n]} \frac{1- F_0(x_k)}{1- F_1(x_k)}
\end{align*}
\noindent since $x_i \leq 1-Q$ and $\frac{1 - F_0(x_i)}{1- F_1(x_i)} = 1$.
\begin{align*}
L \frac{x_j}{1 - x_j} &= \prod_{k \in [n], k \neq j} \frac{1- F_0(x_k)}{1- F_1(x_k)} = \frac{1- F_1(x_j)}{1- F_0(x_j)} \prod_{k \in [n]} \frac{1- F_0(x_k)}{1- F_1(x_k)} 
\end{align*}

Combining the above
\begin{align*}
L \frac{x_j}{1 - x_j} \leq L \frac{1-Q}{Q} \frac{1- F_1(x_j)}{1- F_0(x_j)}
\end{align*}

Or
\begin{align*}
\frac{1-x_j}{x_j} \frac{1-F_1(x_j)}{1-F_0(x_j)} \geq \frac{Q}{1-Q}
\end{align*}

\noindent which is \eqref{start-cascade} for $x = x_j$. Therefore delegation implies that cascades can start, and all players have thresholds $\leq \frac{1}{2}$, proving items \eqref{can-delegate} and \eqref{half}.

The delegation is always by the earliest players, if a delegation strategy always has the best discriminator specified in \eqref{una-discriminator}. Clearly this is true if it is true in any two player game, so let two players play $\mathcal{S}(x_1)$ and $\mathcal{S}(x_2)$, with the first a delegation strategy $x_1 \leq 1 - Q$, and the second not $x_2 > 1 - Q$. Then
\begin{align*}
D(x_1) = L [1 - F_1(x_2)]^2 - [1 - F_0(x_2)]^2 \\
D(x_2) = L [1 - F_1(x_1)]^2 - [1 - F_0(x_1)]^2 = L - 1
\end{align*}

Then, delegating is by the first player if
\begin{align}
\label{d12}
L [1 - F_1(x_2)]^2 - [1 - F_0(x_2)]^2 > L - 1
\end{align}

From \eqref{una-threshold}
\begin{align*}
L \frac{x_2} {1 - x_2} &= \frac{1 - F_0(x_1)}{1 - F_1(x_1)} = 1
\end{align*}

I.e., $L = \frac{1 - x_2}{x_2}$. Substituting in \eqref{d12}
\begin{align*}
\frac{1 - x_2}{x_2} [1 - F_1(x_2)]^2 - [1 - F_0(x_2)]^2 > \frac{1 - x_2}{x_2} - 1
\end{align*}

Or
\begin{align*}
(1 - x_2) [1 - F_1(x_2)]^2 - x_2 [1 - F_0(x_2)]^2 > 1 - 2 x_2
\end{align*}

\noindent proving item \eqref{time-reverse}.

Similarly, delegating is always by the last player when
\begin{align*}
(1 - x_2) [1 - F_1(x_2)]^2 - x_2 [1 - F_0(x_2)]^2 < 1 - 2 x_2
\end{align*}

\noindent proving item \eqref{time-forward}.

\end{proof}

\section{Proof of Proposition \ref{equal}}

\begin{proof}
Let  $i$ and $j$ play be two players in an $n$-unanimity game with prior likelihood $L$. 

Define $$M := L \frac{\prod_{k=1}^n \frac{1 - F_1(x_k)}{1 - F_0(x_k)}}{\frac{1 - F_1(x_i)}{1 - F_0(x_i)}\frac{1 - F_1(x_j)}{1 - F_0(x_j)}}$$
Substituting in \eqref{threshold}, $x_i$ and $x_j$ satisfy the system of equations

\begin{align*}
\left\{  \begin{array}{ll}
M \frac{x_i}{1 - x_i} &= \frac{1 - F_0(x_j)}{1 - F_1(x_j)} \\
M \frac{x_j}{1 - x_j} &= \frac{1 - F_0(x_i)}{1 - F_1(x_i)}
\end{array} \right.
\end{align*}

Multiplying these two equations
\begin{align*}
M \frac{x_i}{1 - x_i} \frac{1 - F_0(x_i)}{1 - F_1(x_i)} = M \frac{x_j}{1 - x_j} \frac{1 - F_0(x_j)}{1 - F_1(x_j)}
\end{align*}

But this means $J(x_i) = J(x_j)$, which, if $J(x)$ is strictly monotonic, implies $x_i = x_j$. 
In a range where a positive function is strictly monotonic, the derivate of its logarithm does not change sign.
\end{proof}

%
%
%
%

\section{Proof of Proposition \ref{large-n}}
\begin{proof}
Let $z := x - (1 - Q)$. From \eqref{bound} we have $z = o(1)$. $f_{\bm{q}}(x)$ is differentiable in the neighborhood below $Q$, so let it's Taylor's expansion there be $$f_{\bm{q}}(x) = a + b (Q - x) + o((Q - x)^2)$$ where $a = f_{\bm{q}}(Q)$ and $b = -f'_{\bm{q}}(Q)$, where $f'_{\bm{q}}(Q)$ denotes the left derivative of $f_{\bm{q}}(x)$ at $Q$. Then, unless $a = b = 0$
\begin{align}
F_0(x) &= \int\limits_{1 - Q}^x (1 - y) f_{\bm{q}}(y) dy = \int\limits_0^z Q(1 \pm o(1)) [a + by + O(y^2)] dy \nonumber \\
\label{F0} &= Q[az + \frac{b}{2}z^2][1 \pm o(1)] \\
F_1(x) &= \int\limits_{1 - Q}^x y f_{\bm{q}}(y) dy = \int\limits_0^z (1-Q)(1 \pm o(1)) [a + by + O(y^2)] dy \nonumber \\
\label{F1} &= (1-Q)[az + \frac{b}{2}z^2][1 \pm o(1)] 
\end{align}

Define $C := \frac{1}{L}\frac{Q}{1-Q} $. From \eqref{bound}, as both sides of the inequality converge to $1$ when $n \to \infty$
\begin{align}
\label{bound1}
\frac{1-F_1(x)}{1-F_0(x)} = C^{\frac{1}{n-1}} [1 \pm o(1)] 
\end{align}

We treat two cases:
\begin{itemize}
\item $a > 0$:

Then terms with $z^2$ may be neglected, and we have, from \eqref{F0} and \eqref{F1}
\begin{align}
\frac{1-F_1(x)}{1-F_0(x)} &= \frac{1 - (1-Q)az}{1 - Qaz} [1 \pm o(1)] \nonumber \\
\label{a>0}
&= 1 + (2Q - 1)az [1 \pm o(1)] 
\end{align}

Equating \eqref{a>0} with \eqref{bound1}, we get
\begin{align*}
z = \frac{C^{\frac{1}{n-1}} - 1}{a (2Q - 1)} [1 \pm o(1)]
\end{align*}

\noindent which proves item (\ref{lg1}) of the Proposition.

We have 
$$\pi_0(L, n, n) = [1 - F_0(x)]^n = (1 - Qaz - O(z^2))^n = (1 - \frac{Q}{2Q -1} [C^{\frac{1}{n-1}} - 1] - O(z^2))^n$$
Noting that, for positive $s, t$ $$\lim\limits_{n \to \infty} [1 - s(t^{\frac{1}{n}} - 1)]^n = t^{-s}$$ we get
$$\pi_0(L, n, n) = C^{-\frac{Q}{2Q -1}} [1 \pm o(1)]$$ which proves item (\ref{lg3}) of the Proposition for the case $a > 0$.

A similar treatment for $\pi_1(L, n, n)$ proves item (\ref{lg4}) of the Proposition for the case $a > 0$.

\item $a = 0, b > 0$:

Then we have, from \eqref{F0} and \eqref{F1}
\begin{align}
\frac{1-F_1(x)}{1-F_0(x)} &= \frac{1 - (1-Q)\frac{b}{2}z^2}{1 - Q\frac{b}{2}z^2} [1 \pm o(1)] \nonumber \\
\label{a=0}
&= 1 + (2Q - 1))\frac{b}{2}z^2 [1 \pm o(1)] 
\end{align}

Equating \eqref{a=0} with \eqref{bound1}, we get
\begin{align*}
z = \sqrt{\frac{C^{\frac{1}{n-1}} - 1}{\frac{b}{2} (2Q - 1)}} [1 \pm o(1)]
\end{align*}

\noindent which proves item (\ref{lg2}) of the Proposition.

In this case, we have 
$$\pi_0(L, n, n) = [1 - F_0(x)]^n = (1 - Q\frac{b}{2}z^2 - O(z^3))^n = (1 - \frac{Q}{2Q -1} [C^{\frac{1}{n-1}} - 1] - O(z^3))^n$$
which leads to the same conclusion as for the case $a > 0$. Similarly for $\pi_1(L, n, n)$. This proves items (\ref{lg3}) and(\ref{lg4}) for the case $a =0$.
\end{itemize}
\end{proof}

\end{document}